\documentclass[aps,reprint,prx,superscriptaddress,a4paper,floatfix]{revtex4-2}

\pdfoutput=1

\usepackage[T1]{fontenc} 
\usepackage[english]{babel} 
\usepackage[utf8]{inputenc}

\usepackage[table,dvipsnames]{xcolor}
\usepackage{hyperref}
\PassOptionsToPackage{linktocpage}{hyperref} 

\hypersetup{
  hypertexnames=false,
  colorlinks   = true, 
  urlcolor     = magenta, 
  linkcolor    = blue, 
  citecolor   = green 
}
\usepackage{bookmark} 

\usepackage{graphicx}
\usepackage{amsthm} 
\usepackage{amssymb}
\usepackage{amsmath,mathtools}
\usepackage{thmtools}
\usepackage{etoolbox}
\usepackage{dsfont}
\usepackage[vcentermath]{youngtab}
\usepackage{comment}
\usepackage{paralist}
\usepackage{xparse}
\usepackage[normalem]{ulem}

\usepackage[capitalize]{cleveref} 
\crefname{figure}{Figure}{Figures} 

\usepackage[
nolist]{acronym}
\makeatletter
\AtBeginDocument{%
  \renewcommand*{\AC@hyperlink}[2]{%
    \begingroup
      \hypersetup{hidelinks}%
      \hyperlink{#1}{#2}%
    \endgroup
  }%
}
\makeatother

\newcommand{\e}{\ensuremath\mathrm{e}} 
\renewcommand{\i}{\ensuremath\mathrm{i}} 
\DeclareMathOperator{\Tr}{Tr} 
\DeclareMathOperator{\rank}{rank} 

\DeclareMathOperator{\CPT}{CPT} 

\DeclareMathOperator{\U}{U}
\newcommand{\id}{\ensuremath\mathrm{id}}
\newcommand{\fro}{\mathrm{F}}

\DeclareMathOperator*{\EE}{\mathbb{E}}

\newcommand{\RR}{\mathbb{R}}

\newcommand{\ZZ}{\mathbb{Z}}

\newcommand{\1}{\mathds{1}} 
\newcommand{\PP}{\mathbb{P}} 
\newcommand{\FF}{\mathbb{F}}

\newcommand{\mc}[1]{\mathcal{#1}}
\renewcommand{\H}{\mc{H}} 

\newcommand{\argdot}{{\,\cdot\,}} 
\renewcommand{\vec}[1]{#1} 

\DeclareMathOperator{\LandauO}{\mc{O}} 

\DeclarePairedDelimiterX{\abs}[1]{\lvert}{\rvert}{%
  \ifblank{#1}{\,\cdot\,}{#1}
}   

\DeclarePairedDelimiterX\norm[1]\lVert\rVert{%
  \ifblank{#1}{\,\cdot\,}{#1}
}   

\newcommand{\lpnorm}[2][p]{\norm{#2}_{\ell_{#1}}}   

\newcommand{\pnorm}[2][p]{\norm{#2}_{#1}} 

\DeclarePairedDelimiterX{\iiiNorm}[1]{\lvert}{\rvert}{%
  \delimsize\lvert\delimsize\lvert#1\delimsize\rvert\delimsize\rvert%
}

\DeclarePairedDelimiterXPP\snorm[1]{}\lVert\rVert{_\infty}{\ifblank{#1}{\,\cdot\,}{#1}}   

\DeclarePairedDelimiterXPP\twonorm[1]{}\lVert\rVert{_2}{\ifblank{#1}{\,\cdot\,}{#1}}   

\DeclarePairedDelimiterXPP\trnorm[1]{}\lVert\rVert{_1}{\ifblank{#1}{\,\cdot\,}{#1}}   

\DeclarePairedDelimiterXPP\fnorm[1]{}\lVert\rVert{_{\fro}}{\ifblank{#1}{\,\cdot\,}{#1}}   

\DeclarePairedDelimiterXPP\dnorm[1]{}\lVert\rVert{_\diamond}{\ifblank{#1}{\,\cdot\,}{#1}}   

\DeclarePairedDelimiterXPP\cbnorm[1]{}\lVert\rVert{_\mathrm{cb}}{\ifblank{#1}{\,\cdot\,}{#1}}   
\DeclarePairedDelimiterXPP\onenorm[1]{}\lVert\rVert{_{1}}{\ifblank{#1}{\,\cdot\,}{#1}}   
\DeclarePairedDelimiterXPP\ddnorm[1]{}\lVert\rVert{_{\diamond\rightarrow \diamond}}{\ifblank{#1}{\,\cdot\,}{#1}}   
\DeclarePairedDelimiterXPP\ssnorm[1]{}\lVert\rVert{_{\infty\rightarrow\infty}}{\ifblank{#1}{\,\cdot\,}{#1}}   

\providecommand\given{}
\newcommand\SetSymbol[1][]{%
  \nonscript\:#1\vert
  \allowbreak
  \nonscript\:
  \mathopen{}}
\DeclarePairedDelimiterX\Set[1]\{\}{%
  \renewcommand\given{\SetSymbol[\delimsize]}
  #1
}

\DeclarePairedDelimiterX\innerp[2]{\langle}{\rangle}{%
  \ifblank{#1}{\,\cdot\,}{#1} , \ifblank{#2}{\,\cdot\,}{#2}%
}

\DeclarePairedDelimiterX\av[1]{\langle}{\rangle}{%
  \ifblank{#1}{\,\cdot\,}{#1}%
}

\DeclarePairedDelimiter{\bra}{\langle}{\vert}
\DeclarePairedDelimiter{\ket}{\vert}{\rangle}

\DeclarePairedDelimiterX\braket[2]{\langle}{\rangle}%
  {#1\kern0.15ex\delimsize\vert\kern0.15ex\mathopen{}#2}

\DeclarePairedDelimiterX\ketbra[2]{\vert}{\vert}%
  {#1\kern0.15ex\delimsize\rangle\delimsize\langle\kern0.15ex\mathopen{}#2}

\DeclarePairedDelimiterX\sandwich[3]{\langle}{\rangle}%
  {#1\,\delimsize\vert\kern0.15ex\mathopen{}#2\kern0.15ex\delimsize\vert\kern0.15ex\mathopen{}#3}

\DeclarePairedDelimiter{\brar}{(}{\vert}
\DeclarePairedDelimiter{\ketr}{\vert}{)}

\DeclarePairedDelimiterX\braketr[2]{(}{)}%
  {#1\kern0.15ex\delimsize\vert\kern0.15ex\mathopen{}#2}

\DeclarePairedDelimiterX\ketbrar[2]{\vert}{\vert}%
  {#1\kern0.15ex\delimsize)\delimsize(\kern0.15ex\mathopen{}#2}

\DeclarePairedDelimiterX\sandwichr[3]{(}{)}%
  {#1\,\delimsize\vert\kern0.15ex\mathopen{}#2\kern0.15ex\delimsize\vert\kern0.15ex\mathopen{}#3}

\NewDocumentCommand\Cl{}{
    \ensuremath{\mathrm{Cl}}%
} 

\NewDocumentCommand\HW{}{
    \ensuremath{\mathcal{P}}%
} 

\newcommand{\normalize}[1]{\check{#1}}
\newcommand{\n}[1]{\normalize{#1}}
\newcommand{\estimator}[1]{\hat{#1}}

\newcommand{\CNOT}{\mathrm{CNOT}}
\newcommand{\CNOTcoset}{\mathcal{N}}
\newcommand{\SWAP}{\mathrm{SWAP}}
\newcommand{\CZ}{\mathrm{CZ}}

\newcommand{\Var}{\mathbb{V}}
\newcommand{\Pg}{\mathsf{P}_n}
\newcommand{\kkbrar}[3]{\ketr{#1 \otimes #2}\brar{#3}}
\newcommand{\ii}{\1}
\newcommand{\one}{\1}
\newcommand{\onen}{\normalize\1}

\newcommand{\Zsf}{\mathsf{Z}}

\DeclareMathOperator{\supp}{supp}

\DeclarePairedDelimiterX\stnorm[1]{\lVert}{\rVert_{\mathrm{st}}}{%
  \ifblank{#1}{\,\cdot\,}{#1}
}   

\newcommand{\m}[1]{\mathcal #1}
\newcommand{\bias}[2]{\abs{\EE[\hat #1] - \av{#2}}}
\newcommand{\sn}{\normalize{\sigma}} 
\newcommand{\s}{\sigma} 

\DeclarePairedDelimiterX\symplecticp[2]{[}{]}{%
  \ifblank{#1}{\,\cdot\,}{#1} , \ifblank{#2}{\,\cdot\,}{#2}%
}

\newcommand{\eff}{\mathrm{eff}}
\newcommand\RS{\mathrm{RS}}

\newtheorem{thm}{Theorem}
\newtheorem{lem}[thm]{Lemma}
\newtheorem{prop}[thm]{Proposition}
\newtheorem{obs}[thm]{Observation}

\definecolor{martin}{rgb}{0,.4,1}

\definecolor{rb}{rgb}{.1,.7,0}
\newcommand{\rb}[1]{{\color{rb}#1}}

\definecolor{markus}{HTML}{006600}

\newcommand{\hhu}{%
  Heinrich Heine University D{\"u}sseldorf, 
  Faculty of Mathematics and Natural Sciences,
  D{\"u}sseldorf, 
  Germany
}

\newcommand{\tuhh}{%
    Hamburg University of Technology, 
    Institute for Quantum Inspired and Quantum Optimization, 
    Hamburg,   
    Germany
}
\newcommand{\tii}{%
Quantum Research Center, Technology Innovation Institute, Abu Dhabi, UAE
}

\begin{document}
\title{Stability of classical shadows under gate-dependent noise}

\author{Raphael Brieger}
\affiliation{\hhu}
\affiliation{\tuhh}
\author{Markus Heinrich}\thanks{Present address: Institute for Theoretical Physics, University of Cologne, Cologne, Germany}
\affiliation{\hhu} 
\author{Ingo Roth}
\affiliation{\tii} 
\author{Martin Kliesch}
\affiliation{\tuhh}

\begin{abstract}
Expectation values of observables are routinely estimated using so-called classical shadows---the outcomes of 
randomized bases measurements on a repeatedly prepared quantum state. 
In order to trust the accuracy of shadow estimation in practice, it is crucial to understand the behavior of the estimators under realistic noise.
In this work, we prove that any shadow estimation protocol involving Clifford unitaries is stable under \emph{gate-dependent} noise for observables with bounded \emph{stabilizer norm}---originally introduced in the context of simulating Clifford circuits. 
In contrast, we demonstrate with concrete examples that estimation of ``magic'' observables can lead to highly misleading results in the presence of miscalibration errors and a worst case bias scaling exponentially in the system size.
We further find that so-called \emph{robust shadows}, aiming at mitigating noise, can introduce a large bias in the presence of gate-dependent noise compared to unmitigated classical shadows.
Nevertheless, we guarantee the functioning of robust shadows for a more general noise setting than in previous works. 
On a technical level, we identify average noise channels that affect shadow estimators and allow for a more fine-grained control of noise-induced biases.
\end{abstract}

\maketitle
 \hypersetup{
       pdftitle = {Stability of classical shadows under gate-dependent noise},
       pdfauthor = {Raphael Brieger, Markus Heinrich, Ingo Roth, Martin Kliesch}
       pdfsubject = {Quantum computing},
       pdfkeywords = {readout, frame, quantum, computer, robust, shadow, estimation, stable, gate, dependent, noise, 
       state, preparation, measurements, errors, SPAM, mitigation, stabilizer, norm, bias, Clifford, variance, sample, complexity, stable, instability, circuits, magic, implementation
       }
      }

Efficient estimation of observables is crucial for quantum experiments and devices. 
Classical shadows \cite{Huang2020Predicting} utilize measurements in randomized bases to perform many relevant estimation tasks on states that are repeatedly prepared in an experiment. 
The approach is highly flexible and has many applications \cite{elbenMixedstateEntanglementLocal2020,zhang2021ExperimentalQuantumState,struchalinExperimentalEstimationQuantum2021,Huang21ProvablyEfficient,Elben22RandomizedMeasurementToolbox,huggins2022unbiasing}.

The central step in the estimation is inverting ``the overall measurement process'' classically. 
This inversion can be calculated analytically for measurement bases that are uniformly random, either local or global, Clifford rotations of the computational bases. 
For these cases,  
tight sampling complexity bounds in terms of certain norms of the observables have been derived \cite{Huang2020Predicting}. 
Most extensions of this paradigm still involve Clifford gates, such as random Clifford circuits \cite{akhtar_scalable_2023,bertoniShallowShadowsExpectation2022,Arienzo22Closed-formAnalytic} or Clifford matchgates \cite{wan2022matchgate}.
In these cases, the inversion is performed using a combination of analytical and numerical techniques.

These calculations crucially rely on the assumption that the unitary gates that perform the basis rotations are perfectly implemented on the quantum device---an assumption that inevitably needs to be relaxed when using classical shadows for precise estimation in practice.
This has motivated several works \cite{Chen21RobustShadowEstimation,KohGrewal:2022:ClassicalShadows,Berg:2020ibi,2023arXiv230504956J,Vitale:2023few,2023arXiv231003071Z,2023arXiv231012726W}
studying noise robust estimators for classical shadows.
Using a restricted noise model, it has been shown that 
the effect of noise on the estimator can be either estimated independently, e.g.\ by a separate calibration experiment \cite{Chen21RobustShadowEstimation} or inferred using symmetries of the prepared state \cite{2023arXiv231003071Z}. 
Once the effect is known, it can then be mitigated in postprocessing.
The derivations of these robust classical shadows assume that the noise in the system is described by the \emph{same} channel acting directly before the measurement in each round. 
Although this restricted noise model is well suited to capture readout noise affecting the measurement \rb{\cite{Vitale:2023few}},
it is difficult to justify such a model for gate noise in realistic experimental setups.
A better-motivated gate-dependent noise model can however significantly complicate the mitigation \cite{Brieger21CompressiveGateSet}. 
It is an open question of how stable classical shadows and their robust extensions are under gate-dependent noise. 
As a matter of concern, the inversion of the effective measurement process typically involves factors that scale exponentially with the system size or locality of the observable.
Thus, even small errors in the gates could induce 
large biases in the estimators. 

In this work, we give explicit analytical and numerical examples where this is indeed the case. 
For instance, measuring a single-qubit state fidelity with over- and underrotated gates can lead to state fidelity errors that are an order of magnitude larger than the average error rates of the gates.
This raises serious concerns about the accuracy of shadow estimation in theory and practice.
On the bright side, we prove that shadow estimation with Clifford circuits is intrinsically stable under gate-dependent noise for observables with bounded \emph{stabilizer norm}.
This includes  stabilizer states, Pauli observables, and large classes of linear combinations thereof.
The \emph{stabilizer norm} \cite{campbell_catalysis_2011}, also known as \emph{$\frac12$-stabilizer R\'enyi entropy} \cite{leone_stabilizer_2022}, is a well-known resource measure in the resource theory of magic states and can be used to bound the runtime of classical stabilizer-based simulation methods \cite{seddon_quantifying_2021}.
For uniform sampling from the local and global Clifford group, we further show that the sampling complexity of shadow estimation is also stable under gate-dependent noise for the same class of observables.

For 
robust classical shadows 
\cite{Chen21RobustShadowEstimation}, we find that the average noise channels can conspire to cause 
an exponentially large bias in the robust estimator, while the bias in the standard estimator is zero.
Taking a closer look at the necessary structure for this situation to appear, we show the stability of the robust shadow estimator 
for a strictly more general gate-dependent noise model than in Ref.~\cite{Chen21RobustShadowEstimation}, which we call \emph{isotropic Pauli noise}. 

\emph{Stability of shadow estimation.}---The goal of a general shadow estimation protocol is to estimate expectation values of observables in the same unknown $n$-qubit quantum state $\rho$.
To this end, the protocol applies a randomly drawn unitary $g$ from a set $G$ with probability $p(g)$ to an unknown quantum state $\rho$ and measures in the computational basis measurement, resulting in some output $x\in\FF_2^n$. 
For an observable $O$, one can then evaluate a function $\estimator{o}(g,x)$ defining an unbiased estimator for the expectation value $\EE[\estimator{o}] = \Tr(O\rho)$.
For this to work for any observable $O$, the operators $\{ g \ketbra{x}{x} g^\dagger \}_{g,x}$ have to form an informationally complete \ac{POVM}.

To give more details, we first introduce some notation.
We define $E_x \coloneqq \ketbra{x}{x}$ and use $\omega(g)$ to denote the unitary channel $\omega(g): A \mapsto g A g^{\dagger}$.
Round brackets are used to denote inner and outer products of linear operators in analogy to the usual Dirac notation.
In particular, $\braketr{A}{B} = \Tr(A^\dagger B)$ denotes the Hilbert Schmidt inner product and $\ketbrar{A}{B}$ is the superoperator $C \mapsto \braketr{B}{C} A$.
With this notation, we define the measure-and-prepare channel $M \coloneqq \sum_{x \in \FF_2^n}\ketbrar{E_x}{E_x}$ in the computational basis.
Being an informationally complete \ac{POVM}, the operators $\{ g E_x g^\dagger \}_{g,x}$ are a frame, i.e.~a spanning set for the vector space of linear operators. Hence the associated \emph{frame operator}
\begin{equation}
S \coloneqq \EE_{g \sim p}\left[\omega(g)^{\dagger} M \omega(g)\right]
\end{equation}
is invertible.
With this notation, we define the estimator $\estimator{o}(g,x) \coloneqq \sandwichr{O}{S^{-1}\omega(g)^\dagger}{E_x}$ of a given observable $O$, and 
a straightforward calculation shows that, indeed, 
$\EE[\estimator{o}] = \sandwichr{O}{S^{-1}S}{\rho} = \braketr{O}{\rho} = \Tr(O\rho)$. 

The frame operator $S$ can often be analytically calculated and inverted for uniform sampling from certain subgroups $G\subset\U(d)$, where $d=2^n$ \cite{Huang2020Predicting}.
A prominent example for such a subgroup is the \emph{Clifford group} $\Cl_n$, defined as the subgroup of $\U(d)$ that is generated by the Hadamard gate, the phase gate, and the controlled-NOT gate.

Gates in an actual experiment, however, suffer from noise and imperfections.
A fairly general and common noise model replaces $\omega(g)$ by its noisy implementation $\phi(g) = \omega(g) \Lambda(g)$, where $\Lambda(g)$ is an arbitrary noise channel that depends on $g$.
Note that introducing an additional noise channel on the left of $\omega$ is equivalent to our model.
The existence of such an implementation map $\phi$ requires the noise to be \emph{Markovian} and \emph{time-stationary}, but it can be otherwise arbitrary.
The \emph{noisy frame operator} of a shadow estimation protocol is then given by 
\begin{equation}\label{eq: Definition noisy frame operator}
\tilde{S} \coloneqq \EE_{g \sim p}\m [\omega(g)^{\dagger} M \m \omega(g) \Lambda(g)]\,.
\end{equation}

In the presence of noise, the standard shadow estimator is biased.
This can be readily seen for a traceless observable $O_0$ and uniform sampling from the Clifford group $\Cl_n$, for which $S^{-1}(O_0) = (d+1)O_0$.
The expected value then reads:
$
\EE[ \estimator o_0 ]
 = \brar{O_0}S^{-1}\tilde S\ketr{\rho}
 = \braketr{O_0}{\rho} + (d+1) \sandwichr{O_0}{S-\tilde S}{\rho}.
$
Because of the exponentially large factor $d+1$ applied in postprocessing, one runs the risk of dramatically amplifying errors. 
In particular, a first straightforward attempt at controlling the noise-induced bias yields a bound of the following form (c.f.~\Cref{app: Noise averaging and bounds}): 

\begin{equation}\label{eq:naive-bound}
\bias{o}{O} 
\leq (d+1) \max_{g \in G} \dnorm{\id - \Lambda(g)}\, .
\end{equation}
Here, we quantify the error of the implementation by the maximum diamond distance of the noise channel to the identity channel over all gates. 
\Cref{eq:naive-bound} is the first example of a bound of the form
\begin{equation*}
  \text{bias} \leq \kappa(d) \times \text{implementation error}.
\end{equation*}
We say that the estimation is \emph{stable} if scaling function $\kappa$ is constant, i.e.\ $\kappa \in \LandauO(1)$. 
However, \cref{eq:naive-bound} suggests that shadow estimation can in fact be \emph{unstable}. 
Indeed, we can give the following example: 

\begin{restatable}{prop}{propsaturation}\label{prop: 1-norm bound saturation}
  Let $O = (\ketbra{H}{H})^{\otimes n}$ with the magic state $\ket{H} = \frac{1}{\sqrt{2}}(\ket{0} + \e^{\i \pi /4}\ket{1})$ and consider shadow estimation with local Clifford unitaries. There exists a state $\rho$ and implementation map $\phi_{\epsilon}(g) = (1-\epsilon)\omega(g) + \epsilon \omega(g) \Lambda(g)$ such that $\bias oO = \kappa\, \epsilon$ with $\kappa \in \Omega(d^{1/4})$.
\end{restatable}
We prove the proposition in \cref{app: Noise averaging and bounds} by an explicit construction.
The construction uses noise channels $\Lambda(g)$ that are coherently undoing the basis change of $\omega(g)$.

This result brings us to the central question of our work: 
\emph{are there 
any classes of stable shadow estimation settings?}
Perhaps surprisingly, we can answer this question positively for large classes of observables in the case that all unitaries are taken from the Clifford group.

In this setting, a careful analysis of the noisy frame operator 
\eqref{eq: Definition noisy frame operator}
allows us to improve the naive estimates significantly.
To this end, it is convenient to work in the orthogonal operator basis given by the \emph{Pauli operators} $\s_a = \s_{a_1}\otimes\dots\otimes\s_{a_n}$, i.e.~$\braketr{\s_a}{\s_{a'}} = d \, \delta_{a,a'}$.
We choose to label them by binary vectors $a \in \FF_2^{2n}$ with the convention $\s_{00} = \1$, $\s_{01} = X$, $\s_{11} = Y$, $\s_{10} = Z$. 
The following technical lemma serves as a basis for our main results.

\begin{restatable}{lem}{noiseavg}\label{lem: Noise averaging}
Suppose $G \subset \Cl_n$. 
The noisy frame operator 
\eqref{eq: Definition noisy frame operator}
takes the form
\begin{align}\label{eq: Avg_noise}
  \tilde{S} = \frac{1}{d} \sum_{a \in \FF_2^{2n}} s_a \ketbrar{\s_a}{\s_a} \bar{\Lambda}_a ,
\end{align}
where $s_a\in (0,1]$ are the eigenvalues of the noise-free frame operator $S$ and $\bar \Lambda_a$ are quantum channels depending on $\Lambda$. 
Furthermore, if the $\Lambda(g)$ are Pauli noise channels, then $\bar\Lambda_a \ketr{\s_a} = \bar\lambda_a \ketr{\s_a}$,
where $\bar \lambda_a \in [-1,1]$.
\end{restatable}

A proof of the lemma is given in \Cref{app: Noise averaging and bounds}.
The channels $\bar \Lambda_a$ are averages of gate-dependent noise channels $\Lambda(g)$ for $g$ sampled from specific subsets of $G$ and approach the identity for weak noise. 
Their form for the global and local Clifford groups is further discussed in \Cref{app: Clifford average structure}. 
In analogy to the channel averages $\bar \Lambda_a$, the parameters $\bar \lambda_a$ are averaged Pauli eigenvalues of individual Pauli noise channels $\Lambda(g)$. 
The proofs builds on the simple observation that in Eq.~\eqref{eq: Definition noisy frame operator}, each term $\omega^{\dagger}(g)M \omega(g)$ is always a sum over Pauli projectors $d^{-1}\ketbrar{\s_a}{\s_a}$, and noise channels $\Lambda(g)$ associated to the same projector can be grouped together. 

Our main stability result uses the \emph{stabilizer norm} of the observable $O$ \cite{campbell_catalysis_2011}, defined as
\begin{align}
  \stnorm{O} \coloneqq \frac{1}{d} \sum_{a \in \FF_2^{2n}}  \abs{ \braketr{\s_a}{O}}\, ,
  \label{eq:def-stab-norm}
\end{align}
which is proportional to the $\ell_1$-norm of an operator in the Pauli basis.
We will also use the \emph{Hilbert-Schmidt norm} $\twonorm{O} \coloneqq \sqrt{\braketr{O}{O}}$ of an operator $O$.  
We now provide a proof that shadow estimation of observables with constant stabilizer norm is stable against gate-dependent noise.

\begin{restatable}{thm}{mainthm}\label{thm: Error bound}
  Suppose $G \subset \Cl_n$.  
  The estimation bias 
  is bounded by
  \begin{equation}\label{eq:result1}
  \begin{split}\bias{o}{O} 
  &\leq \stnorm{O} \max_{a \in \FF_2^{2n}} \dnorm{\id - \bar \Lambda_a} \\
  &\leq \stnorm{O} \max_{g \in \Cl_n} \dnorm{\id - \Lambda(g)}
  \end{split}
  \end{equation} for arbitrary gate-dependent noise and 
  \begin{equation}\nonumber
    \bias{o}{O} \leq \min\left\{\twonorm{O},\stnorm{O}\right\} \max_{a \in \FF_2^{2n}}\abs{1- \bar \lambda_a}
  \end{equation}
    for gate-dependent Pauli noise.
\end{restatable}
\begin{proof}
We obtain 
$S^{-1}\tilde S = d^{-1} \sum_{a}\ketbrar{\s_a}{\s_a} \bar \Lambda_a$ from \cref{lem: Noise averaging}.
Expanding
$O$ in the Pauli basis,
the bias $|\langle O \rangle - \brar{O}S^{-1}\tilde S \ketr{\rho}| = |\brar{O} (\id - S^{-1}\tilde S) \ketr{\rho}|$ becomes:
\begin{equation*}
  \begin{split}
  |\brar{O} (\id - S^{-1}\tilde S) \ketr{\rho}| 
  &= \frac{1}{d} \Big|\sum_a \braketr{O}{\s_a} \brar{\s_a}\id - \bar \Lambda_a\ketr{\rho}\Big| \\
  &\leq \max_a \dnorm{\id - \bar \Lambda_a} \times \frac{1}{d} \sum_a |\braketr{O}{\s_a}|\, , 
  \end{split}
\end{equation*}
  where we used 
  the matrix Hölder inequality, 
  $|\brar{\s_a}\id - \bar \Lambda_a\ketr{\rho}| \leq \dnorm{\id - \bar \Lambda_a}$.
  The result \eqref{eq:result1} then follows from the definition \eqref{eq:def-stab-norm} of the stabilizer norm.
  For Pauli noise, \Cref{lem: Noise averaging} implies that $S^{-1}\tilde S = d^{-1}\sum_{a} \ketbrar{\s_a}{\s_a} \bar \lambda_a$. 
  The bias then becomes 
  $
  \frac{1}{d} |\brar{O} \sum_{a} \ketbrar{\s_a}{\s_a}(1 - \bar \lambda_a) \ketr{\rho}| 
  \leq \max_a |1-\bar \lambda_a| \frac{1}{d} \sum_{a \neq 0} |\braketr{O}{\s_a} \braketr{\s_a}{\rho}|.
  $
  We can bound the last term by either 
  $\max_a|1-\bar \lambda_a| \twonorm{O}$
 via the Cauchy-Schwarz inequality or using $|\braketr{\s_a}{\rho}| \leq 1$,
 by $d^{-1} \sum_{a \neq 0} \braketr{O}{\s_a} = \stnorm{O}$.
\end{proof}

\Cref{thm: Error bound} holds for any informationally complete shadow estimation protocol based on Clifford gates, including uniform sampling from the global or local Clifford group, as well as for alternative proposals such as brickwork circuits \cite{Arienzo22Closed-formAnalytic,bertoniShallowShadowsExpectation2022}.
In comparison to \cref{eq:naive-bound}, the scaling factor $d+1$ is replaced by the stabilizer norm $\stnorm{O}$.
The latter is a magic measure \cite{campbell_catalysis_2011}, equivalent to the $\frac12$-stabilizer R\'enyi entropy \cite{leone_stabilizer_2022}, and can be understood as a quantification of the ``nonstabilizernes'' of the observable $O$.
In particular, $\stnorm{O} = \LandauO(1)$ for many interesting examples such as Pauli observables or stabilizer states.
Furthermore, $\stnorm{O}$ is also well-behaved for nonstabilizer observables, as long as their Pauli support (and coefficients) is not too large.
For instance, consider a $k$-local Hamiltonian $\sum_{e\in E}h_{e}$ on a hypergraph $(V,E)$ with maximum degree $D$, i.e.~each $h_{e} = \sum_{a} h_{e}^a \s_a$ is a linear combination of at most $3^k$ Pauli matrices supported on the hyperedge $e$. 
Thus, the number of local terms is $|E| \leq |V|D = nD$, and assuming $|h_{e}^a| \leq 1$, the stabilizer norm satisfies $\stnorm{H} \leq nD 3^k$.
Finally, the stabilizer norm is 
multiplicative under tensor products
and, hence, simple to compute for product observables, but may be harder to evaluate for more general $O$.
To this end, it might be worth mentioning that $\stnorm{O}$ can be upper bounded by higher stabilizer R\'enyi entropies \cite{leone_stabilizer_2022} and the so-called robustness of magic \cite{2017PhRvL.118i0501H}.

Interestingly, the classical postprocessing of shadow estimation involves the evaluation of $\estimator{o}(g,x) = \sandwichr{O}{S^{-1}\omega(g)^\dagger}{E_x}$.
This evaluation can be a computationally hard problem,
even for uniform sampling from the (local or global) Clifford group where the difficult bit reduces to evaluating expectation values on stabilizer states.
This computational task is well studied and can be solved in classical runtime $\LandauO(\stnorm{O}^2)$ \cite{Rall2019SimulationPaulPropagation,seddon_quantifying_2021}.
Therefore, a compelling observation is that \emph{observables with bounded magic allow for both efficient classical post-processing and stable estimation}.

Compared to the arbitrary noise case, the Pauli noise bound in \cref{thm: Error bound} is both stronger in the error measure and in the dependence on the observable.
Concerning the latter, the bound by $\twonorm{O}$ is favorable when $O$ is a quantum state or an entanglement witness, since then $\twonorm{O} \leq 1$.
The error scaling for Pauli noise is also advantageous, since we derive $\max_a |1-\bar\lambda_a| \leq \max_a \dnorm{\id-\bar\Lambda_a}$ (see \cref{lem:pauli-channels} in \cref{app: pauli-channels}).

Finally, another question that arises in regard to \Cref{thm: Error bound} is whether the given bounds are tight, especially whether the error can really scale with $\stnorm{O}$. 
If so, we would essentially recover our naive bound in \cref{eq:naive-bound} for highly magic observables with $\stnorm{O} = 2^{\LandauO(n)}$.
In fact, the already discussed example in \cref{prop: 1-norm bound saturation} also saturates the bound in \Cref{thm: Error bound} since $\stnorm{\ketbra{H}{H}^{\otimes n}} = \stnorm{\ketbra{H}{H}}^{n} = \left( \frac{1+\sqrt{2}}{2} \right)^n \geq 2^{n/4}$. 

A potential instability could also arise if the protocol's \emph{sample complexity} is no longer controlled due to the effects of gate-dependent noise.
The sample complexity is governed by the estimator variance, which is given in the noise-free scenario and for uniform sampling in Ref.~\cite{Huang2020Predicting}.
Among others, these results involve the  observable's spectral norm $\snorm{O}$.
In the presence of noise, we can only recover a similar result for observables with $\stnorm{O} = \LandauO(\snorm{O})$:

\begin{restatable}{thm}{thmVariance} \label{thm: Variances}
The variance of shadow estimation for uniform sampling from the global Clifford group under gate-dependent noise is bounded by
\begin{equation*}
  \Var_{\mathrm{global}}[\estimator{o}] \leq \frac{2(d+1)}{(d+2)} \stnorm{O_0}^2  + \frac{d+1}{d}\twonorm{O_0}^2 \, .
\end{equation*}
For uniform sampling from the local Clifford group, the variance is bounded by $\Var_{\mathrm{loc}}[\estimator{o}] \leq 4^k \norm{O_\mathrm{loc}}_{\infty}^2$ for $k$-local observables $O = O_{\mathrm{loc}}\otimes \ii^{n-k}$ and by $\Var_{\mathrm{loc}}[\estimator{o}] \leq 3^{\mathrm {supp}(\sigma_a)}$ for Pauli observables $O = \sigma_a$. 
\end{restatable}{}
Our findings indicate that the sample complexity of the protocol for highly magical observables with $\stnorm{O} = 2^{\LandauO(n)}$ may no longer be controlled. 
Intriguingly, we exactly recover the noise-free variance bound of Ref.~\cite{Huang2020Predicting} for the local Clifford group, meaning that the sample complexity for many important use cases is not adversely affected by gate-dependent noise.  
The proof of \cref{thm: Variances} is given in \cref{app: variance bound}, including partial results on more general sampling schemes.


\emph{Robustness of bias mitigation.}---Modifications to the original shadow estimation protocol, coined \emph{\ac{RSE}} \cite{Chen21RobustShadowEstimation}, have already been proposed with the goal of mitigating noise-induced biases.
Their performance guarantees rely on the assumption of gate-independent (left) noise and their success under more general noise models is unclear.
Intuitively, one may hope that \ac{RSE} mitigates the ``dominant contributions'' of an otherwise complicated noise model and therefore generally improves the shadow estimate.

In numerical simulations \cite{Chen21RobustShadowEstimation,2023arXiv231003071Z}, simple gate-dependent noise models are found to be sufficiently well behaved in order for \ac{RSE} to reduce the estimation bias.
In fact, the \ac{RSE} technique and variations thereof have been successfully applied in several experiments \cite{Vitale:2023few,huggins2022unbiasing,Scheurer:2023pgd}. 
For shadows with single-qubit rotations on a superconducting qubit platform, Ref.~\cite{Vitale:2023few} finds that the measurement noise is still dominated by local noise in the readout and, thus, is well captured by left gate-independent noise. 
Similarly using superconducting qubits, Ref.~\cite{huggins2022unbiasing} finds ratios fidelity estimators to be intrinsically robust in practice, which is theoretically expected for left gate-independent noise or readout noise.

We demonstrate, however, 
that \ac{RSE} can \emph{increase the bias rather than reduce it} when 
the noise model assumption is violated. 
The bias of \ac{RSE} can even scale as 
\begin{equation} \label{eq: RSE bias}
  \bias{o_\mathrm{RS}}{O} \geq \abs{\av{O_0}(\frac{1}{2}(1 + \epsilon)^n-1 )}
\end{equation}
for local gate-dependent noise $\phi_{\epsilon}(g) = (1-\epsilon) \omega(g) + \epsilon\, \omega(g)\Lambda(g)$.
As we show, this situation generally requires the effect of the noise to be aligned with the support of the observable or the state under scrutiny.
To explain these findings, we investigate the influence of gate-dependent Pauli noise on \ac{RSE} in more 
detail, using the general form of the frame operator given in \cref{lem: Noise averaging}. 
In particular, we identify a strictly more general noise model than gate-independent left noise \cite{Pauli-noise-assumption}
for which \ac{RSE} works as intended and the bias is strongly suppressed in the number of qubits.
We call this gate-dependent noise model \emph{isotropic Pauli noise}.
Here, the average Pauli eigenvalues $\bar \lambda_a$ are allowed to fluctuate randomly around a mean value in a rotation-invariant fashion in the space of eigenvalues. 
This model is motivated by the intuition that complicated noise processes and the effective averaging introduced by shadow estimation can be well approximated by normally distributed eigenvalues.

We provide an extensive discussion on the stability of \ac{RSE} against gate-dependent noise in \cref{app: robust shadows,app: Error mitigation}.
\begin{figure}[htb] 
  \includegraphics{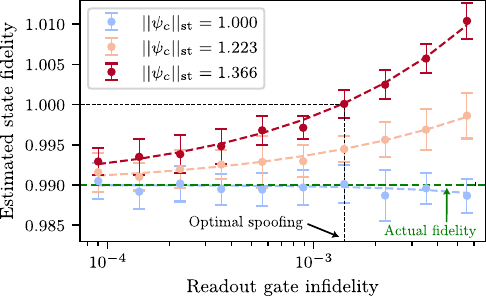}
  \caption{Shadow estimates for the fidelity of single-qubit target states $\ket{\psi_c}$, based on random Pauli basis measurements performed by noisy $\pi/2$ rotations about the $X$ or $Y$ axis. 
  We assume that the actually prepared state $\tilde \rho_c$ is slightly depolarized such that $\sandwich{\psi_c}{\tilde \rho_c}{\psi_c}=0.99$ and consider noisy readout gates $\e^{\i \frac{1}{2}(\frac{\pi}{2} + \delta) X}$ and $\e^{-\i \frac{1}{2}(\frac{\pi}{2} - \delta) Y}$, where $
  \delta$ is a systematic calibration error.
  The solid lines correspond to the exact shadow estimates, while the data points are obtained from simulating the shadow protocol using $10^5$ shots.
  We observe that the estimation bias is approximately one order of magnitude larger than the readout gate infidelity.
  Thus, one can spoof the state fidelity by deliberately making the gates worse, as indicated by the optimal spoofing point in the figure.
  \label{fig: magic fidelity plot}
  }
\end{figure}
\emph{Noise-spoofed classical shadows.}---So far, we have analyzed the 
stability
of shadow estimation through its scaling behavior. 
Here, we demonstrate the identified adverse effects on single-qubit fidelity estimation.
To study the influence of the stabilizer norm, we consider a family of states $\{\ket{\psi_c}\}$, 
tuning with $c \in [0, 1/\sqrt{3}]$ between the $|+\rangle$ state with unit stabilizer norm and the ``facet'' magic state with maximal stabilizer norm (cf.~\cref{app: Spoofing fidelity estimation}).
We assume that a depolarized state $\tilde\rho_c$ is prepared and attempt to estimate its fidelity with $\ket{\psi_c}$.
We randomly apply $\pi/2$ rotations about the $X$ or $Y$ axis to measure in a random Pauli basis.
As a simple, but practically plausible gate-dependent noise model, we consider systematic under- and overrotations (e.g.~due to miscalibration).
The results are shown in \cref{fig: magic fidelity plot}, for which we convert the readout gate error into an average gate infidelity as it would have been measured in a single-qubit Clifford RB experiment \footnote{Following the gate-dependent RB literature \cite{Helsen20RBFramework}, we performed this conversion as follows: We compiled any of the 24 single-qubit Cliffords in $X/Y$ $\pi/2$-pulses, and computed the channel twirl under the over-/underrotation noise model with parameter $\delta$.
The second-largest eigenvalue of this map yields the RB decay parameter $r$ which is then converted into the average gate fidelity $f$ using the standard formula $f=r+(1-r)/d$ where $d=2$ is the Hilbert space dimension.}.
For nonstabilizer states, we find that classical shadows systematically overestimate the fidelity, which can even result in values greater than one.
The estimation bias is approximately one order of magnitude larger than the readout gate infidelity, and increases with the stabilizer norm, in accordance with \cref{thm: Error bound}.
Attempting to use robust shadows in this scenario would lead to an even more severe overreporting of the fidelity, which we explore in \cref{app: Spoofing fidelity estimation}.
Unphysical estimates larger than one can even be amplified exponentially by extending this example straightforwardly to $n$ qubits, showing that arbitrarily large fidelity estimates can be achieved due to miscalibration.

\emph{Conclusion and outlook.}---
Overcoming the restrictions of previous work, we develop a theory for classical shadows based on Clifford unitaries under general gate-dependent, time-stationary, and Markovian noise. 
We find that shadow estimation is, perhaps surprisingly, robust 
for large classes of observables, including convex combinations of Pauli observables and stabilizer states: 
the estimation bias scales with the strength of the gate noise, as well as the observable's \emph{stabilizer norm}. 
For observables with bounded stabilizer norm also the sample complexity remains stable under noise.  
Intriguingly, this is also the class of observables for which one expects efficient classical postprocessing using stabilizer techniques. 
In contrast, highly ``magic'' observables can be sensitive to noise, which we show numerically for systematic over- and underrotations---noise that might occur in current experiments.
In the worst case, magic observables can even suffer from an exponentially large bias.
Finally, we show that the intrinsic stability of classical shadows can 
render it more noise-robust than
``robust classical shadows''. 
We extend the regime in which robust shadows work reliably to certain gate-dependent noise models, which we call isotropic Pauli noise. Beyond that, ref.~\cite{Brieger21CompressiveGateSet} develops a robust classical shadows protocol under local gate-dependent noise using gate-set tomography and, very recently, Refs.~\cite{hu2024demonstration, farias2024robust} demonstrate robust shadows with shallow circuits under more fine-grained gate-dependent noise assumptions.

Our results provide guidance in devising further approaches to mitigate noise in classical shadows, provide crucial justification and identify caveats.

\begin{acknowledgments}
This work was funded by the 
  Deutsche Forschungsgemeinschaft (DFG, German Research Foundation) via the Emmy Noether program (Grant No.\ 441423094),
  the German Federal Ministry of Education and Research (BMBF) within the funding program ``Quantum technologies--from basic research to market'' in the joint project MIQRO (Grant No.\ 13N15522), 
  and by the Fujitsu Germany GmbH as part of the endowed professorship ``Quantum Inspired and Quantum Optimization''. 
\end{acknowledgments}


\clearpage
\onecolumngrid

\appendix
\pdfbookmark{Appendix}{Appendix}

\section{The frame operator under gate-dependent right noise} 
\label{app: Noise averaging and bounds}
We draw gates from a set $G$ and aim for the target implementation $\omega(g)(\rho) = g\rho\, g^{\dagger}$ acting as unitaries $g\in \U(d)$ on density operators $\rho$.
We assume that there is a corresponding hardware implementation
$\phi:G\to \CPT(\H)$ that takes each gate $g$ to some \ac{CPT} map $\phi(g)$ on $\H$.

\begin{prop}[Direct worst-case error bound] \label{prop: direct worst case error bound}
Consider an observable $O_0$ that, w.l.o.g., we assume to be traceless. 
The bias of shadow estimation with uniform sampling from a unitary 2-design is bounded as 
\begin{align}
\abs*{\brar{O_0}S^{-1}\tilde{S}\ketr{\rho} - \braketr{O_0}{\rho}} \leq (d+1) \EE_{g \in G} \dnorm{\phi(g) - \omega(g)}\, .
\end{align}
\end{prop}
\begin{proof}
For uniform sampling from a unitary 2-design (such as the Clifford group), the inverse of the ideal frame operator acts as $\brar{O_0} S^{-1} = (d+1) \brar{O_0}$ \cite{Huang2020Predicting}.
Let $\mu$ be the uniform probability measure on $G$, for instance the normalized Haar measure on $\U(d)$ or $\Cl_n$. 
The error due to noise in the frame operator is then given by
\begin{align}
  \left|\brar{O_0}S^{-1}\tilde{S}\ketr{\rho} - \brar{O_0}S^{-1} S\ketr{\rho}\right| 
  &= (d+1) \brar{O_0} \left[\int \omega(g)^{\dagger} M \phi(g) \mathrm{d}\mu(g) - \int \omega(g)^{\dagger} M \omega(g) \mathrm{d}\mu(g)\right] \ketr{\rho}\\
  & \leq (d+1) \dnorm{\int \omega(g)^{\dagger} M \phi(g) \mathrm{d}\mu(g) - \int \omega(g)^{\dagger} M \omega(g) \mathrm{d}\mu(g)}\\
  & = (d+1) \dnorm{\int \omega(g)^{\dagger} M (\phi(g) - \omega(g)) \mathrm{d}\mu(g)} \label{eq: Dia norm bound line 3}\\
  & \leq (d+1) \int \dnorm{\omega(g)^{\dagger} M} \dnorm{\phi(g) - \omega(g)} \mathrm{d}\mu(g) \label{eq: Dia norm bound line 4}\\
  &= (d+1) \int\dnorm{\phi(g) - \omega(g)} \mathrm{d}\mu(g)\,,
\end{align}
where we used the submultiplicativity of the diamond norm and that $\dnorm{\mc C} = 1$ for any quantum channel $\mc C$, in particular for $\mc C = M$.
\end{proof}
For trace preserving $\phi(g)$, one can quickly verify that $|\brar{O}S^{-1}\tilde{S}\ketr{\rho} - \braketr{O}{\rho}| = |\brar{O_0}S^{-1}\tilde{S}\ketr{\rho} - \braketr{O_0}{\rho}|$, where $O_0$ is the traceless component of $O$.
This bound suggests an error amplification by a dimensional factor $d+1$. 
However, the bound cannot be tight since the triangle inequality from \cref{eq: Dia norm bound line 3} to \cref{eq: Dia norm bound line 4} can only be saturated in the trivial zero-error case $\phi(g) = \omega(g)$. 
Indeed, in the following, we show that much stronger bounds for sampling measurement unitaries from the Clifford group can be derived. 
The defining property of Clifford unitaries that they map the Pauli group onto itself plays a central role: 
even in the presence of noise, the frame operator can be written in a more amenable form when looking at its Pauli transfer matrix.

In the following, we denote by $\Pg$ the set of $n$-qubit Pauli operators.
The action of Clifford group elements on $M$ via the representation $\omega$ is given by 
\begin{equation}
\label{eq:clifford-action}
  \omega(g)(\s_a) = (-1)^{\varphi_a(g)} \Xi_a(g)\,,
\end{equation}
with functions $\varphi_a: \Cl_n \rightarrow \FF_2$ and $\Xi_a: \Cl_n \rightarrow \Pg$ defined for $a \in \FF_2^{2n}$.
Moreover, we write $\sn_a = \s_a/\sqrt{d}$ for the normalized Pauli operators such that $\braketr{\sn_a}{\sn_b} = \delta_{a,b}$.
In the following, we denote by $Z_z \equiv \bigotimes_{i = 1}^n Z^{z_i}$ with $\vec z \in \FF_2^n$ the diagonal Pauli operators and set $Z_1 \equiv Z_{e_1}$ with $e_1 = (1,0,\dots,0)$.
We write again $\normalize Z_z = Z_z / \sqrt{d}$ for their normalized versions.
Moreover, one can easily verify that $M$ is given in terms of the diagonal Paulis as
\begin{equation}
\label{eq:MZ}
  M = \sum_{x \in \FF_2^n}\ketbrar{E_x}{E_x} = \sum_{z \in \FF_2^n} \ketbrar{\n Z_z}{\n Z_z}\, .
\end{equation}
\noiseavg*

\begin{proof}
We show the statement by a direct calculation: 
\begin{align}
\tilde S &=  \sum_{g \in G} p(g) \omega(g)^{\dagger} \sum_{z \in \FF_2^n} \ketbrar{\n Z_z}{\n Z_z} \omega(g) \Lambda(g) \\
&= \sum_{g \in G}  \sum_{z \in \FF_2^n} p(g) (-1)^{\varphi_a(g)} \ketbrar{\normalize\Xi_z(g)}{\normalize\Xi_z(g)}  (-1)^{\varphi_a(g)}\Lambda(g)  \\
&= \sum_{g \in G}  \sum_{z \in \FF_2^n} p(g) \ketbrar{\normalize\Xi_z(g)}{\normalize\Xi_z(g)}  \Lambda(g)  \\
&= \sum_{a \in \FF_2^{2n}} \ketbrar{\sn_a}{\sn_a}   \sum_{z \in \FF_2^n} \sum_{g \in \Xi_z^{-1}(\s_a)} p(g) \Lambda(g) \\
&= \sum_{a \in \FF_2^{2n}} s_a \ketbrar{\sn_a}{\sn_a}   \sum_{z \in \FF_2^n} \sum_{g \in \Xi_z^{-1}(\s_a)} \frac{p(g)}{s_a} \Lambda(g) \\
&= \sum_{a \in \FF_2^{2n}} s_a \ketbrar{\sn_a}{\sn_a} \bar{\Lambda}_a \, ,
\end{align} 
where $s_a \coloneqq  \sum_{z \in \FF_2^n} \sum_{g \in \Xi_z^{-1}(s_a)} p(g)$ and  $\bar \Lambda_a \coloneqq \sum_{z \in \FF_2^n} \sum_{g \in \Xi_z^{-1}(\s_a)} \frac{p(g)}{s_a} \Lambda(g)$.\\
For Pauli noise parameterized by $\Lambda(g) = \ketbrar{\onen}{\onen} + \sum_{a \neq 0}\lambda_a(g)\ketbrar{\sn_a}{\sn_a}$, the term $\brar{\sn_a}\bar \Lambda_a$ simplifies to
\begin{align} \label{eq-app: Average Pauli channel eigenvalues}
\brar{\sn_a}\bar \Lambda_a = \sum_{z \in \FF_2^n}\sum_{g \in \Xi_z^{-1}(\s_a)} \frac{p(g)}{s_a} \brar{\sn_a}\Lambda(g) = \sum_{z \in \FF_2^n}\sum_{g \in \Xi_z^{-1}(\s_a)} \frac{p(g)}{s_a} \lambda_a(g) \brar{\sn_a} \eqqcolon \bar \lambda_a \brar{\sn_a}\,.
\end{align}
The frame operator for Pauli noise then becomes $\tilde{S} = \sum_{a \in \FF_2^{2n}} s_a \bar{\lambda}_a \ketbrar{\sn_a}{\sn_a}$.
One can quickly verify that the set $\Set{\frac{p(g)}{s_a}\given g \in \Xi_z^{-1}(\s_a), z\in \FF_2^n}$ corresponds to a normalized probability distribution by definition of $s_a$. 
Thus, each $\bar{\Lambda}_a$ is an average over noise channels and thereby, a quantum channel itself. The condition $\bar \lambda_a \in [-1,1]$ follows from the general property of Pauli eigenvalues $\lambda_a(g) \in [-1,1]$.
\end{proof}

Crucially, these average channels depend on $a$ and are taken over preimages of $\s_a$. 
The most immediate example where \Cref{lem: Noise averaging} can be applied is shadow estimation with the global Clifford group for uniform sampling, i.e.\ $p(g) = 1/|\Cl_n|$. 
In this case, we have $s_a = 1/(d+1)$ for $a \neq 0$ and $s_0 = 1$ \cite{Huang2020Predicting}. 
Furthermore, each $\bar \Lambda_a$ is an average over $|\Cl_n|/(d+1)$ right noise channels $\Lambda(g)$, and a characterization of the sets over which averages are taken can be found in \cref{app: Clifford average structure}. 

For trace-preserving noise, the identity component in $\tilde S$ can be treated differently since, for any channel $\Lambda(g)$, the trace-preservation condition is equivalent to $\brar{\ii}\Lambda(g) = \brar{\ii}$ and, thus, $\brar{\ii} \bar{\Lambda}_a = \brar{\ii}$. 
Moreover, since the adjoint unitary action $\omega(g)$ for each gate $g$ is trace-preserving, it holds that $\brar{\ii}\omega(g)\ketr{Z_z} = \braketr{\ii}{Z_z} = d \, \delta_{0,z}$. 
As $\Xi_z^{-1}(\ii)$ is defined to be precisely the set of all gates that map $Z_z$ to $\ii$, it turns out that $\Xi_z^{-1}(\ii) = \emptyset$ for $z\neq 0$ and $\Xi_z^{-1}(\ii) = \Cl_n$ for $z = 0$. 
Thus, for any probability distribution $p(g)$ over $\Cl_n$, it holds true that $s_{0} = \sum_{z \in \FF_2^n}\sum_{g \in \Xi_z^{-1}(\ii)} p(g) = \sum_{g \in \Xi_{0}^{-1}(\ii)} p(g) = \sum_{g \in \Cl_n} p(g) = 1$. 

For clarity of presentation, we restate the main Theorem and its proof here.

\mainthm*

\begin{proof}
  From \Cref{lem: Noise averaging} we see that the ideal frame operator with $\Lambda(g) = \id$ for any $g \in G$ is given by 
  \begin{align}
  S = \sum_{a \in \FF_2^{2n}} s_a \ketbrar{\sn_a}{\sn_a} \sum_{z \in \FF_2^n}\sum_{g \in \Xi_z^{-1}(\s_a)} \frac{p(g)}{s_a} \, \id = \sum_{a \in \FF_2^{2n}} s_a \ketbrar{\sn_a}{\sn_a} \,,
  \end{align}
   and its inverse is $S^{-1} = \sum_{a \in \FF_2^{2n}} \frac{1}{s_a} \ketbrar{\sn_a}{\sn_a}$. Therefore the effective operation performed by the shadow estimation protocol is
  $S^{-1} \tilde S = \sum_{a \in \FF_2^{2n}} \ketbrar{\sn_a}{\sn_a} \bar \Lambda_a$.
  We can now bound the absolute error as follows: 
  \begin{align}
  \abs[\big]{\braketr{O}{\rho} - \brar{O}S^{-1}\tilde S \ketr{\rho}} 
  &= 
  \abs[\Big]{\brar{O} \Big(\id - \sum_a \ketbrar{\sn_a}{\sn_a} \bar{\Lambda}_a \Big)\ketr{\rho}} \label{eq:Bound proof 1}\\
  &= \abs[\Big]{\sum_a \braketr{O}{\sn_a} \brar{\sn_a}(\id - \bar \Lambda(a))\ketr{\rho}} \label{eq:Bound proof 2}\\
  &\leq \sum_a |\braketr{O}{\sn_a}| \,\,|\brar{\sn_a}(\id - \bar \Lambda_a\ketr{\rho}| \label{eq:Bound proof 3}\\
  &\leq \sum_a |\braketr{O}{\sn_a}|\, \,  \snorm{\sn_a} \pnorm[1]{(\id -  \bar \Lambda_a)(\rho)} \label{eq:Bound proof 4} \\
  &\leq \sum_a |\braketr{O}{\sn_a}|\, \, \snorm{\sn_a} \dnorm{\id - \bar \Lambda_a} \label{eq:Bound proof 5}\\
  &\leq \max_a \dnorm{\id - \bar \Lambda_a} \sum_a \frac{|\braketr{O}{\sn_a}|}{\sqrt{d}} \label{eq:Bound proof 6}\\
  &= \max_a \dnorm{\id - \bar \Lambda_a} \stnorm{O} \,, 
  \end{align}
where 
from \eqref{eq:Bound proof 3} to \eqref{eq:Bound proof 4} we used the matrix Hölder inequality.
From \eqref{eq:Bound proof 5} to \eqref{eq:Bound proof 6} we used $\snorm{\sn_a} = \snorm{\s_a}/\sqrt{d} = 1/\sqrt{d}$ and $\pnorm[1]{(\id -  \bar \Lambda_a)(\rho)} \leq \dnorm{\id - \bar \Lambda_a}$ \cite{Wat18}.
Finally, we used $\braketr{O}{\sn_a} = \braketr{O}{\s_a}/\sqrt{d}$ and the definition of the stabilizer norm in \cref{eq:def-stab-norm}.

It remains to prove the bound for Pauli noise,
where we use that $\brar{\sn_a}\bar{\Lambda}_a = \bar \lambda_a \brar{\sn_a}$ as seen in \cref{eq-app: Average Pauli channel eigenvalues}. 
This leads to $S^{-1} \tilde S = \sum_{a \in \FF_2^{2n}} \bar \lambda_a \ketbrar{\sn_a}{\sn_a}$ and $\id - S^{-1} \tilde S = \sum_{a \in \FF_2^{2n}} (1 - \bar \lambda_a) \ketbrar{\sn_a}{\sn_a} $. 
Consequently, we can bound the bias for Pauli noise as
\begin{align}
\abs[\big]{\braketr{O}{\rho} - \brar{O}S^{-1}\tilde S \ketr{\rho}} 
&= 
\Big|\brar{O} \sum_{a \in \FF_2^{2n}} (1 - \bar \lambda_a) \ketr{\sn_a} \braketr{\sn_a}{\rho}\Big|\label{eq: Pauli noise bound 1}\\
&\leq \max_a\abs{1-\lambda_a} \sum_{a \neq 0} |\braketr{O}{\sn_a} \braketr{\sn_a}{\rho}| \label{eq: Pauli noise bound 2}\\
&\leq \max_a\abs{1-\lambda_a} \twonorm{O} \twonorm{\rho} \label{eq: Pauli noise bound 3}\\
&\leq \max_a\abs{1-\lambda_a} \twonorm{O} \label{eq: Pauli noise bound 4}\, .
\end{align}
From Eq.~\eqref{eq: Pauli noise bound 2} to Eq.~\eqref{eq: Pauli noise bound 3}, the Cauchy-Schwarz inequality has been used, and the last line follows from the fact that 
$\twonorm{\rho}\leq 1$.
We can alternatively bound Eq.~\eqref{eq: Pauli noise bound 2} by using $|\braketr{\sn_a}{\rho}|\leq 1/\sqrt{d}$ to come by $\max_a\abs{1-\lambda_a} \sum_{a} |\braketr{O}{\sn_a} \braketr{\sn_a}{\rho}| \leq \max_a\abs{1-\lambda_a} \sum_{a} |\braketr{O}{\sn_a}|/\sqrt{d} =  \max_a\abs{1-\lambda_a} \stnorm{O}$.
\end{proof}

As mentioned in the discussion of \Cref{thm: Error bound} in the main text, the error bound for Pauli noise is as strong as one could hope for in the case where $O$ has unit rank. 
The maximum error on Pauli eigenvalues is bounded by the diamond norm for Pauli channels and $\twonorm{O}$ relates to the largest expectation value over all input states via $\twonorm{O} \leq \sqrt{\rank O}\pnorm[\infty]{O}$. 
In general we have $\stnorm{O} 
\leq \sqrt{d} \twonorm{O}$, meaning that the bound in terms of $\twonorm{O}$ can be stronger by a factor of $\sqrt{d}$ compared to the bound in terms of $\stnorm{O}$. 
We will now give an explicit example showing that, at least for local Clifford shadow estimation, the bias can scale exponentially in the system size. 

\propsaturation*

\begin{proof}
In the following we use a noise model $\Lambda(g) = \bigotimes_{i = 1}^n \Lambda_i(g_i)$.
Recall from \cref{eq:clifford-action} that the local Cliffords $g_i$ act as $\omega(g_i)(Z) = (-1)^{\varphi_Z(g_i)} \Xi_Z(g_i)$.
Then, we define
\begin{align}\label{eq: worst case error definition}
\Lambda_i(g_i) =  \begin{cases}
  \mc \omega(g_i)^{\dagger}\omega(X)  & \text{if }\varphi_Z(g_i) = 1 \\
  \omega(g_i)^{\dagger} & \text{if }\varphi_Z(g_i) = 0\, ,
\end{cases}
\end{align}
where $\omega(X)(\rho) = X \rho X$. The frame operator inherits the $\epsilon$-dependence from $\phi_{\epsilon}$, and we write $\tilde S_{\epsilon} \equiv \tilde S(\phi_{\epsilon})$.

Let us first consider the case $\epsilon = 1$.
Since we can write the Z-basis measurement operator as $M = \sum_{z \in \FF_2^n}\ketbrar{\n Z_z}{\n Z_z} = (\ketbrar{\onen}{\onen} + \ketbrar{\normalize Z}{\normalize Z})^{\otimes n}$, the frame operator factorizes for the above local noise model:
\begin{align}
\tilde S_1 &= \frac{1}{|\Cl_1^{\times n}|} \sum_{g \in \Cl_1^{\times n}} \omega(g)^{\dagger} \sum_{z \in \FF_2^n} \ketbrar{\n Z_z}{\n Z_z} \omega(g)\Lambda(g) \\
& = \bigotimes_{i = 1}^n \left(\frac{1}{|\Cl_1|} \sum_{g_i \in \Cl_1} \omega(g_i)^{\dagger} \left(\ketbrar{\onen}{\onen} + \ketbrar{\normalize Z}{\normalize Z}\right) \omega(g_i)\Lambda_i(g_i) \right).
\end{align}
The action of a local Clifford operation $g_i$ can be written as $\omega(g_i)^{\dagger}\ketr{Z} = (-1)^{\varphi_Z(g_i)}\ketr{\Xi_Z(g_i)}$ with $\varphi_Z(g_i) \in \FF_2$. 
By inserting $\Lambda_i(g_i)$ we get
\begin{align}
\tilde S_1 & = \bigotimes_{i = 1}^n \left( \ketbrar{\onen}{\onen} + \frac{1}{|\Cl_1|} \sum_{g_i \in \Cl_1} (-1)^{\varphi_Z(g_i)}\ketbrar{\normalize\Xi_Z(g_i)}{\normalize Z} (-1)^{\varphi_Z(g_i)} \right) 
= \left( \ketbrar{\onen}{\onen} + \frac{1}{3} \sum_{a \in \FF_2^2\setminus\{0\}} \ketbrar{\sn_a}{\normalize Z} \right)^{\otimes n}\label{eq: Scaling proof 2}\,.
\end{align}
With this error model, each measurement is done in the Z-basis, and the sign cancels with the sign acquired in post-processing from the action of $\omega(g_i)^{\dagger}$. 
Since the ideal frame operator is given by $S = \tilde S_0 = \bigotimes_{i = 1}^n \left( \ketbrar{\onen}{\onen} + \frac{1}{3} \sum_{a \neq 0} \ketbrar{\sn_a}{\sn_a} \right)$, we find that $S^{-1}\tilde S = \left( \ketbrar{\onen}{\onen} + \sum_{a \neq 0} \ketbrar{\sn_a}{\n Z} \right)^{\otimes n}$. The expectation value of an observable $O$ on the initial state $E_0 = \left[\frac{1}{\sqrt{2}}\left(\n \ii + \n Z\right)\right]^{\otimes n}$ is consequently given by
\begin{align}
  \brar{O}S^{-1}\tilde S_1 \ketr{E_0} &= \frac{1}{\sqrt d}\brar{O} \left(\ketr{\n \ii} + \sum_{a \in \FF_2^2\setminus\{0\}}\ketr{\sn_a}\right)^{\otimes n} \\
  &= \frac{1}{\sqrt d} \sum_{a \in \FF_2^{2n}} \braketr{O}{\sn_a} \, .
  \label{eq:saturation-exp-value}
\end{align}
It remains now to compute $\tilde S_{\epsilon}$ for the implementation map $\phi_{\epsilon}(g) = (1-\epsilon)\omega(g) + \epsilon \omega(g)\Lambda(g)$. For this we note that 
\begin{align}
  \tilde S_{\epsilon} &= \EE_g [\omega^{\dagger}(g)M((1-\epsilon)\omega(g) + \epsilon \omega(g) \Lambda(g))] \\
   &= (1-\epsilon) \EE_g [\omega^{\dagger}(g)M\omega(g)] + \epsilon \EE_g [\omega^{\dagger}(g)M\omega(g)\Lambda(g))] \\
  &= (1-\epsilon) \tilde S_0 + \epsilon \tilde S_1.
\end{align}
Therefore the bias is given by $\bias{o}{O} = \abs{(1-\epsilon)\av{O} + \epsilon \sandwichr{O}{S^{-1}\tilde S_1}{E_0} - \av{O}} = \epsilon \abs{\sum_{a \in \FF_2^{2n}} \braketr{O}{\sn_a}/\sqrt{d} - \av{O}}$. \\
So far the calculation was independent of the observable, and we now consider the magic state $O = (\ketbra{H}{H})^{\otimes n}$. From the definition of $\ket{H}$ we can read off $\av{O} = |\braket{H}{E_0}|^{2n} = \frac{1}{d}$. Furthermore, one can show that $\ketbra{H}{H} = \frac{1}{\sqrt 2}\left[\n \ii + \frac{1}{\sqrt{2}}(\n X + \n Y)\right]$, and we get
\begin{align}
  \sandwichr{O}{S^{-1}\tilde S_1}{E_0} 
  = \left[\frac{1}{2} \sum_{a \in \FF_2^2} \braketr{\brar{\n \ii} + \tfrac{1}{\sqrt{2}}(\brar{\n X} + \brar{\n Y})}{\sn_a}\right]^{n}
  = \left[\frac{1+\sqrt{2}}{2}\right]^n.
\end{align}
Since the entries of $O = (\ketbra{H}{H})^{\otimes n}$ in the Pauli basis are positive, the stabilizer norm of $O$ is by \cref{eq:saturation-exp-value} exactly the just derived expression, $\stnorm{O} = \left[\frac{1+\sqrt{2}}{2}\right]^n \geq 2^{n/4}$. By putting everything together, we arrive at the desired result $\bias{o}{O} = \abs{\stnorm{O} - 1/d}\epsilon = \kappa \epsilon$.
\end{proof}

Whether the same error scaling with $\stnorm{O}$ can also occur for Clifford-based shadow estimation protocols other than uniform sampling from the local Clifford group remains open. 
To understand for instance the difference between local and global Cliffords, we can look at the above example in terms of average error channels. 
What allows the errors to accumulate in \cref{prop: 1-norm bound saturation} is the fact that for local noise, the local noise averages $\bar{\Lambda}_X, \bar{\Lambda}_Y, \bar{\Lambda}_Z$ are each taken over disjoint subsets over the local Clifford group and, thus, an error model exists such that they can be independently chosen. 
For uniform sampling from the global Clifford group this is not the case anymore. As shown in the next section, each of the $d^2 -1$ average noise channels $\bar \Lambda_a$ is an average over $|\Cl_n|/(d+1)$ 
channels $\Lambda(g)$, therefore the $\bar \Lambda_a$ cannot all be independent.

\section{Details on the spoofing of fidelity estimation via miscalibration} 
\label{app: Spoofing fidelity estimation}

We consider the family of pure states $\ketr{\rho_c}$ studied in \citet{seddon_quantifying_2021} defined by
\begin{equation}\label{eq: magic state family parametrization}
  \ketr{\rho_c} = \frac{1}{\sqrt{2}} \left(\ketr{\n \ii} + c\ketr{\n Y} + c\ketr{\n Z} + \sqrt{1 -2c^2}\ketr{\n X}\right).
\end{equation}
Here, $c \in [0,1/\sqrt{3}]$ such that $\rho_0 = \ketbra{+}{+}$ and $\rho_{1/\sqrt{3}} = \ketbra{F}{F}$, where $\ket{F}$ is called the facet state. The stabilizer norm of this state family can be easily determined using the parametrization in \cref{eq: magic state family parametrization}, resulting in 
\begin{equation}
  \stnorm{\rho_c} = (1+2c+\sqrt{1-2c^2})/2
\end{equation}
Furthermore, $\ket{F}$ maximizes the stabilizer norm among all single qubit pure states.
In the setting of \cref{fig: magic fidelity plot}, a shadow estimation protocol with Pauli basis measurements is implemented using the set of Clifford gates $\{g_0 = \ii, g_1 = \e^{\i \frac{\pi}{4} X}, g_3 = \e^{\i \frac{\pi}{4} Y}\}$, where in each round one of the gates is drawn at random and applied before measurement. 
We consider the following unitary noise model: 
\begin{equation}\label{eq:NoiseChannels}
\begin{aligned}
  \Lambda(g_0) &= \id, 
  \\
  \Lambda(g_1)[\rho] &= \e^{\i \frac{\delta}{2} X} \rho\, \e^{-\i \frac{\delta}{2} X}, 
  \\
  \Lambda(g_2)[\rho] &= \e^{-\i \frac{\delta}{2} Y} \rho\, \e^{\i \frac{\delta}{2} Y}\, ,
\end{aligned}
\end{equation}
i.e., a small overrotation along the $X$ axis and a small underrotation along the $Y$ axis on the Bloch sphere. The single qubit frame operator under this noise model is given as 
\begin{equation}\label{eq: Overrotation frame operator}
    \tilde S = \frac{1}{3} \sum_{g \in \{\ii, \e^{\i \frac{\pi}{4} X}, \e^{\i \frac{\pi}{4} Y}\}} \omega^{\dagger}(g)\left( \ketbrar{\n \ii}{\n \ii} + \ketbrar{\n Z}{\n Z}\right)\omega(g)\Lambda(g).
\end{equation}
It is straightforward to verify that 
\begin{equation}
\begin{split}
  \brar{Z} \omega(g_1) &= \brar{\e^{-\i \frac{\pi}{4} X} Z \e^{\i \frac{\pi}{4} X}} = - \brar{Y} , \\
  \brar{Z} \omega(g_2) &= \brar{\e^{-\i \frac{\pi}{4} Y} Z \e^{\i \frac{\pi}{4} Y}} = \brar{X}
  \end{split}
\end{equation}
and
\begin{equation}
\begin{split}
  \brar{Z} \omega(g) \Lambda(g) &= \brar{\e^{- \frac{\i}{2}(\frac{\pi}{2} + \delta) X} Z \e^{\frac{\i}{2} (\frac{\pi}{2} + \delta) X}} = - \cos(\delta) \brar{Y} - \sin(\delta) \brar{Z}, \\
  \brar{Z} \omega(g) \Lambda(g) &= \brar{\e^{-\frac{\i}{2} (\frac{\pi}{2} - \delta) Y} Z \e^{\frac{\i}{2} (\frac{\pi}{2} - \delta) Y}} = \cos(\delta) \brar{X} + \sin(\delta) \brar{Z}.
  \end{split}
\end{equation}
By plugging these expressions into \cref{eq: Overrotation frame operator}, we arrive at 
\begin{equation}\label{eq: Overrotation frame operator final}
    \tilde S = \frac{1}{3} \left(3 \ketbrar{\n \ii}{\n \ii} + \ketr{\n Y}(\cos(\delta)\brar{\n Y} + \sin(\delta) \brar{\n Z}) + \ketr{\n X}(\cos(\delta)\brar{\n X} + \sin(\delta) \brar{\n Z})\right) + \ketbrar{\n Z}{\n Z}.
\end{equation}
Assuming $\ketr{\rho_c}$ was prepared, we can calculate the expected fidelity of being in $\ketr{\rho_c}$ given the noisy frame operator $\tilde S$. A brief calculation reveals that
\begin{equation}
  \sandwichr{\rho_c}{S^{-1} \tilde S}{\rho_c} = \frac{1}{2} \left(1 + \sin(\delta) (c \sqrt{1 - 2c^2} + c^2) + \cos(\delta) (1 - c^2) + c^2\right).
\end{equation} 
For the highest stabilizer norm state in the family ($c = 1/\sqrt{3}$), we find $\sandwichr{\rho_c}{S^{-1} \tilde S}{\rho_c} = \frac{2 + \sqrt{2}\sin(\delta + \pi/4)}{3}$. 
Thus, the expected outcome for the estimated fidelity is larger than one for $\delta \in [0,\pi/2]$. If the implemented state is, for instance, affected by depolarizing noise $\ketr{\tilde \rho_c} = \frac{p}{\sqrt{2}} \ketr{\n \ii} + (1-p) \ketr{\rho_c}$ and hence, 
$\sandwichr{\rho_c}{S^{-1} \tilde S}{\tilde \rho_c} = \frac{p}{2} + (1-p) \sandwichr{\rho_c}{S^{-1} \tilde S}{\rho_c}$, the noise can be masked by an overestimated fidelity due to miscalibration, as shown in 
Figure~1 in the main text.

A similar behavior can be observed when randomly sampling from local Cliffords gates, and the numerical results are shown in \cref{fig: Local Clifford magic}. For this experimental setup, we compile each single qubit Clifford gate
with Bloch rotations $\{g_0 = \e^{\i \frac{\pi}{4} Z}, g_1 = \e^{\i \frac{\pi}{4} X}, g_3 = \e^{\i \frac{\pi}{4} Y}\}$, which we take to be affected by the noise channels 
\begin{equation}
\begin{aligned}
  \Lambda(g_0)[\rho] &= \rho, 
  \\
  \Lambda(g_1)[\rho] &= \e^{\i \frac{\delta}{2} X} \rho\, \e^{-\i \frac{\delta}{2} X}, 
  \\
  \Lambda(g_2)[\rho] &= \e^{-\i \frac{\delta}{2} Y} \rho\, \e^{\i \frac{\delta}{2} Y} \, .
\end{aligned}
\end{equation}
The bias of shadow estimation shown in \cref{fig: Local Clifford magic} is computed with a fixed calibration error $\delta$ for each instance of the protocol. If we compare the results to Figure 1 in the main text, we see that the overestimation problem for the fidelity is not as pronounced as for local Pauli basis measurements. This is due to the randomizing effect of using the whole Clifford group, which leads to less biased average noise channels. Since local Cliffords are compiled using several Bloch rotations per Clifford gate, the error per readout gate is in principle larger in this scenario, which is evident by the steeper decrease in fidelity for the $\stnorm{\psi}=1$ case (blue dashed line). \cref{fig: Local Clifford magic} also shows us a case where robust shadows increase the bias, as can be seen for instance by comparing the dashed red line in both plots. These results underline the main takeaways form \cref{thm: Error bound}, namely that unmitigated shadows are inherently stable under gate dependent noise, \emph{if} the stabilizer norm of the observable is controlled. 

\begin{figure}[ht]
\includegraphics{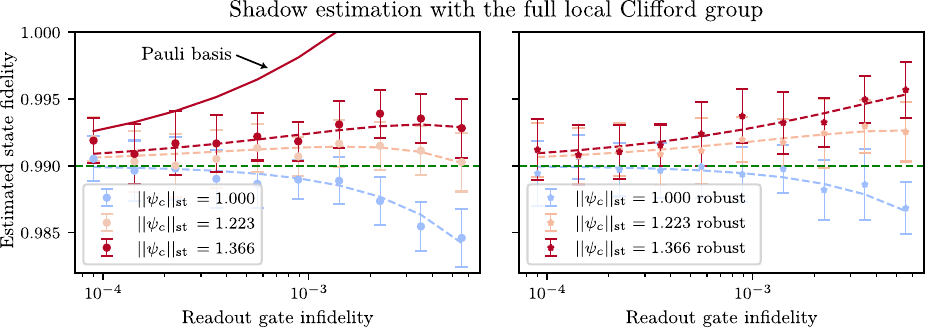}
\caption{
Estimated state fidelity determined by classical shadows using random local Clifford measurements, for 3 states with different stabilizer norms. On the \textbf{left} side, we show the result for unmitigated shadows, while on the \textbf{right} side we show the result for robust shadows. The green line indicates the true state fidelity for all tested states. Data points are the empirical mean of 20 shadow estimates,
and error bars indicate the variance. Each shadow estimate itself is obtained from $10^5$ shots. Dashed lines are the theoretical results, while the continuous red line in the left plot is the theoretical line for the Pauli basis measurements, which were discussed in the main text.
}
\label{fig: Local Clifford magic}
\end{figure}

\section{Explicit calculation of the noise average for the Clifford group} 
\label{app: Clifford average structure}
In the following,
we write $\CNOT_{i,j}$ for the controlled-NOT gate with control qubit $i$ and target qubit $j$.
More generally, given a 2-qubit unitary $U$, we write $U_{i,j}$ for its application on the ordered qubit pair $(i,j)$.
Moreover,
$\SWAP$ is the 2-qubit SWAP gate and $\HW_n = \langle \Pg \rangle$ denotes the $n$-qubit Pauli group.

\begin{lem} \label{lem: Z-Basis measurement decomposition}
  Let 
  $\CNOTcoset\coloneqq \bigcup_{i=1}^{n-1}\left\{\prod_{j = 1}^i U_{j+1,j}| U \in \{\SWAP, 
  \CNOT\}\right\} \cup \ii$.
  Then the following holds: 
  \begin{align}\label{eq: CX description of measurement operator}
  M = \sum_{z \in \FF_2^n}\ketbrar{\normalize Z_z}{\normalize Z_z} = \ketbrar{\onen}{\onen} + \sum_{g \in \CNOTcoset} \omega(g)^{\dagger} \ketbrar{\normalize Z_1}{\normalize Z_1} \omega(g)\, .
  \end{align}
\end{lem}

\begin{proof}
\newcommand{\bigT}{\mbox{\Large $T$}}
To see this we first note that SWAP and $\CNOT_{2,1}$ are given as
\begin{equation*}
\CNOT_{2,1} = 
\begin{pmatrix}
  1 & 0 & 0 & 0 \\
  0 & 0 & 0 & 1 \\
  0 & 0 & 1 & 0 \\
  0 & 1 & 0 & 0 
\end{pmatrix},
\quad 
\mathrm{SWAP} = 
\begin{pmatrix}
  1 & 0 & 0 & 0 \\
  0 & 0 & 1 & 0 \\
  0 & 1 & 0 & 0 \\
  0 & 0 & 0 & 1 
\end{pmatrix}.
\end{equation*}
Since $Z_{(1,0)} = \mathrm{diag}(1,1,-1,-1)$, $Z_{(0,1)} = \mathrm{diag}(1,-1,1,-1)$ and $Z_{(1,1)} = \mathrm{diag}(1,-1,-1,1)$ are all diagonal and the adjoint actions of $\mathrm{SWAP}$ and $\CNOT_{2,1}$ permute diagonal elements, one can quickly verify that $\omega(\mathrm{SWAP})\ketr{Z_{(1,0)}} =  \ketr{Z_{(0,1)}}$ 
and $\omega(\CNOT_{2,1}) \ketr{Z_{(1,0)}} = \ketr{Z_{(1,1)}}$. 
Consecutively applying either $\CNOT_{2,1}$ or SWAP along the qubit chain to $Z_1$ generates all $Z_z$ with $z \in \FF_2^n\backslash\{\vec{0}, \vec{e}_1\}$. 
One can also see that $|\CNOTcoset| = 1 + \sum_{i = 1}^{n-1}2^i = 2^n - 1$, which is the number of non-identity terms in \cref{eq: CX description of measurement operator}.
\end{proof}

\begin{lem}\label{lem: e_1 stabilizer}
The stabilizer of $\ketbrar{Z_1}{Z_1}$ under the action $g\mapsto\omega^{\dagger}(g) (\argdot) \omega(g)$ with $g \in \Cl_n$ is given by 
$\mathrm{St}_{e_1} \equiv \mathrm{St}_{Z_1} = \Cl_{n-1} \cdot \mathrm{\HW_n} \cdot \langle S_1, \{\CZ_{1,i}, \CNOT_{1,i}\}_{i \in \{2, \dots, n\}}\rangle$.
\end{lem}

\begin{proof}
Since $Z_1$ is just the identity on qubits $2, \dots, n$, any adjoint action by a unitary acting only on qubits $2, \dots, n$ leaves $Z_1$ invariant. 
This holds, in particular, for $\id \otimes \omega(g)$ with $g \in \Cl_{n-1}$. 
For $g \in \HW_n$ we know that $\omega(g)^{\dagger}\ketr{\s_a} = \pm \ketr{\s_a}$ for any $a \in \FF_2^{2n}$ and, thus, $\omega^{\dagger}(g)\ketbrar{\s_a}{\s_a} \omega(g) = \ketbrar{\s_a}{\s_a}$.
The remaining Clifford group elements that can stabilize $\ketbrar{Z_1}{Z_1}$ act only on qubit 1 or between qubit 1 and qubits $2, \dots n$. 
Since $Z_1$ is diagonal in the computational basis, this includes all diagonal Cliffords which we did not already count in $\Cl_{n-1}$, namely $\langle S_1, \CZ_{1,i > 1} \rangle$. 
In addition, identical diagonal elements of $Z_1$ can be permuted and one can easily verify that, for instance, $\omega(\CNOT)\ketr{Z_{(1,0)}} = \ketr{Z_{(1,0)}}$. 

To show that this is indeed the complete stabilizer, note that the orbit of $\ketbrar{Z_1}{Z_1}$ corresponds to all non-identity Pauli strings and thus has size $4^n-1$.
Hence, we have $|\Cl_n| / |\mathrm{St}_{e_1}| = 4^n-1 $, where
\begin{align}
|\Cl_n| &= 2^{2n+3} 2^{n^2} \prod_{i = 1}^n (4^i -1), &  |\HW_n| &= 2^{2n+2}\, .
\end{align}
The order of the above defined group $\Cl_{n-1} \cdot \HW_n \cdot \langle S_1, \{\CZ_{1,i}, \CNOT_{1,i}\}_{i \in \{2, \dots, n\}}\rangle$ can be computed by observing that we can simply compute the cardinalities of the first and last factor up to Pauli operators, and multiply those by $|\HW_n|$. 
Since CNOT normalizes diagonal Clifford unitaries, $\langle S_1, \CZ_{1,i}, \CNOT_{1,i}\rangle = \langle S_1, \CZ_{1,i} \rangle \rtimes \langle \CNOT_{1,i}\rangle $ is a semidirect product and $\langle S_1, \CZ_{1,i} \rangle$ is Abelian. 
We have $2^{n-1}$ possibilities of applying $\CZ_{1,i}$ and $\CNOT_{1,i}$ and the S-gate has order 4, hence $|\langle S_1, \CZ_{1,i}, \CNOT_{1,i}\rangle / \langle Z_1 \rangle| = 2 \cdot 2^{n-1} 2^{n-1} = 2^{2n-1}$. 
Putting everything together, we find
\begin{align}
|\Cl_{n-1} \cdot \HW_n \cdot \langle S_1, \{\CZ_{1,i}, \CNOT_{1,i}\}_{i \in \{2, \dots, n\}}\rangle| &= |\Cl_{n-1}/\HW_{n-1}| \times |\HW_n| \times |\langle S_1, \CZ_{1,i}, \CNOT_{1,i}\rangle / \langle Z_1 \rangle | \\
&= 2^{(n-1)^2+1} \prod_{i = 1}^{n-1}(4^i - 1) \times 2^{2n+2} \times 2^{2n-1}\\
&= 2^{2n+3} 2^{n^2} \frac{1}{4^n - 1} \prod_{i = 1}^{n} (4^i - 1) \\
&= \frac{|\Cl_n|}{4^n -1}\, ,
\end{align}
which shows the claim.
\end{proof}

Let $g_a$ be any element of $\Cl_n$ that satisfies $\omega^{\dagger}(g_a)\ketbrar{Z_1}{Z_1}\omega(g_a) = \ketbrar{\s_a}{\s_a}$ for $a\neq \ii$. The stabilizer of $\ketbrar{\s_a}{\s_a}$ can then simply be written as $\mathrm{St}_a = g_a \mathrm{St}_{e_1}$. 
Moreover, the 
stabilizers $g_a \mathrm{St}_{e_1}$ are exactly the left cosets of $\mathrm{St}_{e_1}$ and, thus, $\Cl_n = \bigcup_{a \neq 0} g_a \mathrm{St}_{e_1}$, where $g_a \mathrm{St}_{e_1}$ and $g_{a'} \mathrm{St}_{e_1}$ are pairwise disjoint for $a \neq a'$.

\begin{prop} \label{prop: Global Clifford average structure}
  The frame operator for shadow estimation with uniform sampling from the $n$-qubit Clifford group is given by 
\begin{align}
\tilde S = \ketbrar{\onen}{\onen}  +  \frac{1}{d+1}  \sum_{a \neq 0} \ketbrar{\sn_a}{\sn_a} \EE_{h \in \CNOTcoset}\, \EE_{h' \in \mathrm{St}_{e_1}} \Lambda(h^{-1} h'g_a), 
\end{align}
where $g_a$ can be any element of $\Cl_n$ satisfying $\omega^{\dagger}(g_a)\ketbrar{Z_1}{Z_1}\omega(g_a) = \ketbrar{\s_a}{\s_a}$.
\end{prop}
\begin{proof}
With the use of \Cref{lem: Z-Basis measurement decomposition} and \Cref{lem: e_1 stabilizer} we can successively rewrite the frame operator as follows:
\begin{align}
\tilde S &= \frac{1}{|\Cl_n|} \sum_{g \in \Cl_n} \omega^{\dagger}(g) \sum_{z \in \FF_2^n} \ketbrar{\normalize Z_z}{\normalize Z_z} \omega(g) \Lambda(g) \\
&= \ketbrar{\onen}{\onen} + \frac{1}{|\Cl_n|} \sum_{h \in \CNOTcoset} \sum_{g \in \Cl_n} \omega(hg)^{\dagger} \ketbrar{\normalize Z_1}{\normalize Z_1} \omega(hg) \Lambda(g) \\
&= \ketbrar{\onen}{\onen} + \frac{1}{|\Cl_n|} \sum_{g \in \Cl_n} \omega(g)^{\dagger} \ketbrar{\normalize Z_1}{\normalize Z_1} \omega(g)  \sum_{h \in \CNOTcoset} \Lambda(h^{-1}g) \\
&= \ketbrar{\onen}{\onen} + \frac{1}{|\Cl_n|} \sum_{g \in \Cl_n / \mathrm{St}_{e_1}} \omega(g)^{\dagger} \ketbrar{\normalize Z_1}{\normalize Z_1} \omega(g)  \sum_{h \in \CNOTcoset, h' \in \mathrm{St}_{e_1}} \Lambda(h^{-1} h'g) \\
&= \ketbrar{\onen}{\onen} + \frac{1}{|\Cl_n|} \sum_{a\neq\ii} \ketbrar{\sn_a}{\sn_a} \sum_{h \in \CNOTcoset, h' \in \mathrm{St}_{e_1}} \Lambda(h^{-1} h' g_a). 
\end{align}
Since $|\mathrm{St}_{e_1}| = \frac{|\Cl_n|}{4^n -1}$ and $|\CNOTcoset| = d-1$, we can rewrite the last line in terms of the channel averages as
\begin{align}
\frac{1}{|\Cl_n|} \sum_{h \in \CNOTcoset, h' \in \mathrm{St}_{e_1}} \Lambda(h^{-1} h' g_a) 
&= \frac{1}{(d+1)(d-1)|\mathrm{St}_{e_1}|} \sum_{h \in \CNOTcoset, h' \in \mathrm{St}_{e_1}} \Lambda(h^{-1} h' g_a) \\
&= \frac{1}{d+1} \EE_{h \in \CNOTcoset}\, \EE_{h' \in \mathrm{St}_{e_1}} \Lambda(h^{-1} h'g_a)\, .
\end{align}
\end{proof}
This result ties back to the general form derived in \Cref{lem: Noise averaging} via $\Xi_{b}^{-1}(a) = h_{b}^{-1}\mathrm{St}_{e_1}g_a$ and $s_a = \frac{1}{d+1}$. 
We will now turn to the local Clifford group and determine the factors $s_a$, as well as the compositions of $\bar \Lambda_a$. The result will take a simpler form for $\emph{local noise}$, which we define as noise that factorizes as $\Lambda(g) = \bigotimes_{i = 1}^n \Lambda^{(i)}(g_i)$ on all $g \in \Cl_1^{\times n}$. We also define the support of $a \in \FF_2^{2n}$ as $\abs{\supp(a)} = |\{i \in [n]: a_i \neq 0\}|$.

\begin{prop}\label{prop:local Cliffords}
Let $G(a) \subseteq \mathrm{Cl}_1$ be defined by
\begin{align}
G(a) = 
\begin{cases}
\mathrm{Cl}_1 & a = 0 \\
\mathrm{St}(Z)g_a & a \in (\FF_2^{2*})\, .
\end{cases}
\end{align}
Then the frame operator for the $n$-qubit local Clifford group is given by
\begin{align}
\tilde S = \sum_{\vec a \in \FF_2^{2n}} \frac{1}{3^{|\supp(\vec a)|}} \ketbrar{\sn_a}{\sn_a} \bar \Lambda_a\, ,
\end{align}
where $\bar \Lambda_a = \EE_{g_1 \in G(a_1)} \cdots \EE_{g_n \in G(a_n)} \Lambda(g_1, \dots, g_n)$ for global noise and 
$
\bar \Lambda_a = \EE_{g_1 \in G(a_1)} \Lambda^{(1)}(g_1) \otimes \cdots \otimes \EE_{g_n \in G(a_n)} \Lambda^{(n)}(g_n) 
= \bigotimes_{i = 1}^n \bar \Lambda_{a_i}
$ 
for local noise.
\end{prop}
\begin{proof}
Since the local Clifford group also satisfies the conditions in \Cref{lem: Noise averaging} it remains to show that the averages $\bar \Lambda_a$ take the above form and that the coefficients are $s_a = \frac{1}{3^{|\supp(a)|}}$. We will now explicitly proof the case $n = 2$, from which the result for an arbitrary system size can be straightforwardly generalized. In this case $M = \left(\ketbrar{\onen}{\onen} + \ketbrar{\normalize Z}{\normalize Z}\right)^{\otimes 2} = \sum_{z_1, z_2 \in \FF_2} \ketbrar{\normalize Z_{z_1}}{\normalize Z_{z_1}} \otimes \ketbrar{\normalize Z_{z_2}}{\normalize Z_{z_2}}$.
For $\Cl_1 \times \Cl_1$, the product representation $\omega(g_1, g_2) = \omega(g_1)\otimes \omega(g_2)$, and implementation map $\phi(g_1,g_2) = (\omega(g_1)\otimes \omega(g_2))\Lambda(g_1,g_2)$, the frame operator becomes

\begin{align}\label{Eq: Avg_noise}
\tilde S &= \frac{1}{|\Cl_1|^2} \sum_{g_1,g_2 \in \Cl_1} \sum_{z_1, z_2 \in \FF_2} \left(\omega(g_1)^{\dagger} \ketbrar{\normalize Z_{z_1}}{\normalize Z_{z_1}} \omega(g_1) \otimes \omega(g_2)^{\dagger} \ketbrar{\normalize Z_{z_2}}{\normalize Z_{z_2}} \omega(g_2)\right) \Lambda(g_1,g_2) \\
&= \frac{1}{|\Cl_1|^2} \sum_{g_2 \in \Cl_1} \sum_{a_1 \in \ZZ_4} \sum_{z_2 \in \FF_2} \left(\ketbrar{\sn_{a_1}}{\sn_{a_1}} \otimes  \omega(g_2)^{\dagger} \ketbrar{\normalize Z_{z_2}}{\normalize Z_{z_2}} \omega(g_2)\right) \sum_{z_1 \in \FF_2} \sum_{g_1 \in \Xi^{-1}_{z_1}(a_1)} \Lambda(g_1,g_2) \\
&= \frac{1}{|\Cl_1|^2} \sum_{a_1 \in \ZZ_4} \ketbrar{\sn_{a_1}}{\sn_{a_1}} \otimes \sum_{a_2 \in \ZZ_4} \ketbrar{\sn_{a_2}}{\sn_{a_2}} \sum_{z_1, z_2 \in \FF_2} \sum_{g_1 \in \Xi^{-1}_{z_1}(a_1)} \sum_{g_2 \in \Xi^{-1}_{z_2}(a_2)} \Lambda(g_1,g_2).
\end{align}
Since $\omega(g)^{\dagger} \ketbrar{\one}{\one}\omega(g) = \1$ for all $g \in \mathrm{Cl}_1$ and $\omega(g)^{\dagger} \ketbrar{\s_a}{\s_a}\omega(g) \neq 1$ for all $a \neq 0$ and $g \in \mathrm{Cl}_1$, we know that $\Xi^{-1}_{0}(\ii) = \mathrm{Cl}_n$, $\Xi^{-1}_{0}(\sigma_a) = \emptyset$ as well as $\Xi^{-1}_{a}(\ii) = \emptyset$. Moreover, if we apply \Cref{prop: Global Clifford average structure} to the case $n=1$, we see that $\CNOTcoset = \{\ii\}$ and $\Xi^{-1}_1(\sigma_a) = \mathrm{St}_1 g_a$. The cosets $G(a \neq 0) = \mathrm{St}_1 g_a$ are disjoint for $a \in \{X,Y,Z\}$ and of order $\abs{\mathrm{Cl}_1}/3$. We then get
\begin{align}
\tilde S & = \sum_{a_1, a_2 \in \ZZ_4} \frac{\abs{G(a_1)}}{\Cl_1} \frac{\abs{G(a_2)}}{\Cl_1} \ketr{\sn_{a_1}} \brar{\sn_{a_1}} \otimes \ketr{\sn_{a_2}} \brar{\sn_{a_2}} \EE_{g_1 \in G(a_1)}\, \EE_{g_2 \in G(a_2)} \Lambda(g_1,g_2) \\
& = \sum_{a_1, a_2 \in \ZZ_4} \frac{1}{3^{|\supp(\vec a)|}} \ketr{\sn_{a_1}} \brar{\sn_{a_1}} \otimes \ketr{\sn_{a_2}} \brar{\sn_{a_2}}  \EE_{g_1 \in G(a_1)} \, \EE_{g_2 \in G(a_2)} \Lambda(g_1,g_2).
\end{align}
The last line follows from 
\begin{align}
\frac{\abs{G(a)}}{\abs{\mathrm{Cl}_1}} = 
\begin{cases}
1 & a = 0 \\
1/3 & a \in (\FF_2^{2*}), 
\end{cases}
\end{align}
which we can write as $\frac{\abs{G(a_1)}}{\abs{\mathrm{Cl}_1}} = 3^{-\abs{\supp(a_1)}}$ whereafter $\frac{\abs{G(a_1)}\abs{G(a_2)}}{\abs{\mathrm{Cl}_1}^2} = 3^{-\abs{\supp(\vec a)}}$. For local noise, we get 
\begin{align}
\EE_{g_1 \in G(a_1)} \, \EE_{g_2 \in G(a_2)} \Lambda(g_1,g_2) = \EE_{g_1 \in G(a_1)}  \Lambda^{(1)}(g_1) \otimes \EE_{g_2 \in G(a_2)} \Lambda^{(2)}(g_2) = \bar \Lambda_{a_1} \otimes \bar \Lambda_{a_2}.
\end{align}
The $n$-qubit local Clifford group result then follows.
\end{proof}

A different locality structure of the noise channels $\Lambda(g)$ (other than fully local noise) will lead to the corresponding locality structure on the average noise channels. For instance, if all noise channels factorize along a bipartition of the set of qubits, the average noise channels will inherit this factorization.

\section{A norm inequality for Pauli channels} \label{app: pauli-channels}
In the following we consider Pauli channels that act as $\Lambda(\rho) = \sum_{b \in \FF_2^{2n}} p_b \s_b \rho \s_b$, with $p_b \in [0,1]$ and $\sum_b p_b = 1$. The corresponding superoperators are known to be diagonal in the Pauli basis and to have eigenvalues $\lambda_a = \sum_{a\in\FF_2^{2n}} (-1)^{[a,b]} p_b$ where $[a,b]=0$ if $\s_a$ and $\s_b$ commute and $[a,b]=1$ otherwise.
\begin{lem}
\label{lem:pauli-channels}
Let $\Lambda$ and $\Lambda'$ be Pauli channels, then it holds that 
$\snorm{\Lambda-\Lambda'} \leq \dnorm{\Lambda - \Lambda'}$.
\end{lem}
\begin{proof}
Let $\Lambda$ and $\Lambda'$ be given by the probability distributions $p$ and $p'$, respectively, hence their eigenvalues are $\lambda_a = \sum_{a\in\FF_2^{2n}} (-1)^{[a,b]} p_b$ and $\lambda'_a = \sum_{a\in\FF_2^{2n}} (-1)^{[a,b]} p'_b$. 
Since $\Lambda$ and $\Lambda'$ are both diagonal in the Pauli basis, we have $\snorm{\Lambda-\Lambda'} = \max_a \abs{\lambda_a-\lambda'_a}$. It then follows that $\max_a \abs{\lambda_a-\lambda'_a} = \max_a \abs{\sum_b (-1)^{[a,b]} (p_b - p'_b )} \leq \sum_b |p_b - p'_b| = \dnorm{\Lambda - \Lambda'}$, where the last step uses a well-known relation for the diamond distance of Pauli channels \cite{Magesan2012}.
\end{proof}
Since the channel average over Pauli channels is again a Pauli channel, the above statement holds in particular for $\bar \Lambda_a$ and the identity operation, $\snorm{\id-\bar\Lambda_a} \leq \dnorm{\id-\bar\Lambda_a}$.

\section{Variance bounds for gate-dependent noise} \label{app: variance bound}
To bound the variance of the estimator, $\Var[\estimator{o}] = \EE(\estimator{o}^2) - \EE(\estimator{o})^2$, we compute the second moment
\begin{align}\label{eq: Variance definition}
\EE(\estimator{o}^2) = \brar{O \otimes O} (S^{-1})^{\otimes 2}  \sum_{x \in \FF_2^n} \sum_{g \in G} p(g) \omega(g)^{\dagger \otimes 2} \kkbrar{E_x}{E_x}{E_x}\omega(g) \Lambda(g) \ketr{\rho} \, .
\end{align}
It will come in handy again to express the map $M_3 \coloneqq \sum_{x \in \FF_2^n}\kkbrar{E_x}{E_x}{E_x}$ in the normalized Pauli basis. 
By using that $\braketr{E_x}{\s_a} = 0$ if $\s_a$ is not diagonal and $\braketr{E_x}{Z_z} = (-1)^{x\cdot z}$, one can readily show that 
\begin{align}
M_3 
= \frac{1}{\sqrt{d}} \sum_{z,z' \in \FF_2^n} \kkbrar{\n Z_z}{\n Z_{z'}}{\n Z_{z+z'}}
= \frac{1}{d^2} \sum_{z,z' \in \FF_2^n} \kkbrar{Z_z}{Z_{z'}}{Z_{z+z'}}\, .
\end{align}
Here and in the following, addition of binary vector such as $z+z'$ is within the binary field $\FF_2$, i.e.~to be taken modulo 2.
In the absence of noise, one can show that the relevant operator in \cref{eq: Variance definition} can be written as
\begin{align}
\label{eq:S3 definition}
 S_3 
 \coloneqq 
 \sum_{g \in G} p(g) \omega(g)^{\dagger \otimes 2} M_3 \omega(g)
 = 
 \frac{1}{\sqrt{d}} \sum_{\substack{a, a' \in \FF_2^{2n}: \\ \symplecticp{a}{a'} = 0}} s_{a,a'} 
 \kkbrar{\sn_a}{\sn_{a'}}{\sn_{a+a'}} \,,
\end{align}
for suitable constants $s_{a,a'}\in \RR$.
Under noise, we show that the analogous operator
\begin{align}
\tilde S_3 \coloneqq \sum_{g \in G} p(g) \omega(g)^{\dagger \otimes 2} M_3 \omega(g) \Lambda(g)
\end{align}
can be brought in a similar form to \cref{eq:S3 definition} where the noise enters linearly and from the right.
We then obtain an analogous statement for the variance as for the expectation value in \cref{lem: Noise averaging}:

\begin{lem}
  Consider a shadow estimation protocol with random sampling from the Clifford group
   according to an arbitrary probability distribution $p$ that ensures informational completeness. 
   Let $s_{a,a'}$ be as in \cref{eq:S3 definition} and let $s_a$ and $s_{a'}$ be as in \cref{lem: Noise averaging}.
   Then, the second moment for an observable $O$ and a state $\rho$ can be written as 
  \begin{align} \label{eq: noisy second moment}
  \EE(\estimator{o}^2) = \frac{1}{\sqrt{d}} \sum_{\substack{a, a' \in \FF_2^{2n}: \\ \symplecticp{a}{a'} = 0}} \frac{s_{a,a'}}{s_a s_{a'}} \braketr{O}{\sn_a}\braketr{O}{\sn_{a'}} \brar{\sn_{a+a'}}\bar\Lambda_{a,a'}\ketr{\rho}\, ,
  \end{align}
  where $\bar\Lambda_{a,a'}$ are suitable averages of the gate noise channels $\Lambda(g)$.
\end{lem}

\begin{proof}
For the target implementation of Clifford unitaries, we can again write the action of $\omega^{\dagger}(g) \otimes \omega^{\dagger}(g) (\cdot) \omega(g)$ on $\ketr{Z_z} \otimes \ketr{Z_{z'}} \brar{Z_{z+z'}}$ in terms of $\varphi_a: \Cl_n \rightarrow \FF_2$ and $\Xi_a: \Cl_n \rightarrow \Pg$. It also holds that $\brar{Z_{z+z'}}\omega(g) = \brar{g Z_z g^{\dagger} g Z_{z'}  g^{\dagger}} = \brar{\Xi_z(g) \Xi_{z'}(g)} (-1)^{\varphi_z(g) + \varphi_{z'}(g)}$, where the Pauli operators $\Xi_z(g)$ and $\Xi_{z'}(g)$ commute. Therefore, we get 
\begin{align*}
\omega^{\dagger}(g) \otimes \omega^{\dagger}(g) \kkbrar{Z_z}{Z_{z'}}{Z_{z+z'}} \omega(g) 
&= (-1)^{\varphi_z(g) + \varphi_{z'}(g)} \kkbrar{\Xi_z(g)}{\Xi_{z'}(g)}{\Xi_z(g) \Xi_{z'}(g)} (-1)^{\varphi_z(g) + \varphi_{z'}(g)} \\
&= \kkbrar{\Xi_z(g)}{\Xi_{z'}(g)}{\Xi_z(g) \Xi_{z'}(g)}\, .
\end{align*}
Note that if $\sigma_a = \Xi_z(g)$ and $\sigma_{a'} = \Xi_{z'}(g)$ for suitable $a,a'$ with $[a,a']=0$, then $\sigma_a \sigma_{a'} = (-1)^{\beta(a,a')} \sigma_{a+a'}$ for a suitable binary function $\beta$.
Since $\tilde S_3$ depends only linearly on the right noise channels $\Lambda(g)$ we can proceed in analogy to \Cref{lem: Noise averaging} and rewrite $\tilde S_3$ as 
\begin{align*}
\tilde S_3 &= \sum_{g \in G} p(g) \omega(g)^{\dagger \otimes 2} M_3 \omega(g) \Lambda(g) \\
&= 
\frac{1}{d^2} \sum_{g \in G} p(g) \sum_{z,z' \in \FF_2^n} \ketr{\Xi_z(g)} \otimes \ketr{\Xi_{z'}(g)} \brar{\Xi_z(g) \Xi_{z'}(g)} \Lambda(g) \\
&= 
\frac{1}{d^2} \sum_{\substack{a, a' \in \FF_2^{2n}: \\ \symplecticp{a}{a'} = 0}} 
\kkbrar{\s_a}{\s_{a'}}{\s_a \s_{a'}} \sum_{z,z' \in \FF_2^n} \sum_{g \in \Xi_z^{-1}(a) \cap \Xi_{z'}^{-1}(a')} p(g) \Lambda(g) \\
&= 
\frac{1}{\sqrt{d}} \sum_{\substack{a, a' \in \FF_2^{2n}: \\ \symplecticp{a}{a'} = 0}} s_{a,a'} 
\kkbrar{\sn_a}{\sn_{a'}}{\sn_{a+a'}} \sum_{z,z' \in \FF_2^n} \sum_{g \in \Xi_z^{-1}(a) \cap \Xi_{z'}^{-1}(a')} \frac{p(g)}{r_{a,a'}} \Lambda(g)\, ,
\end{align*}
where $s_{a,a'} = (-1)^{\beta(a,a')} r_{a,a'}$ and $r_{a,a'} = \sum_{z,z' \in \FF_2^n} \sum_{g \in \Xi_z^{-1}(a) \cap \Xi_{z'}^{-1}(a')} p(g)$. 
We can now define average channels $\bar \Lambda_{a,a'} \coloneqq \sum_{z,z' \in \FF_2^n} \sum_{g \in \Xi_z^{-1}(a) \cap \Xi_{z'}^{-1}(a')} \frac{p(g)}{r_{a,a'}} \Lambda(g)$ and write $\tilde S_3$ as 

\begin{align} \label{eq: noisy second moment operator}
\tilde S_3 
= 
\frac{1}{\sqrt{d}} \sum_{\substack{a, a' \in \FF_2^{2n}: \\ \symplecticp{a}{a'} = 0}} 
s_{a,a'} \kkbrar{\sn_a}{\sn_{a'}}{\sn_{a+a'}} \bar \Lambda_{a,a'}\, .
\end{align}
The second moment of $\estimator{o}$ is then
\begin{equation}
  \EE(\estimator{o}^2) = \brar{O \otimes O} S^{-1}\otimes S^{-1} \tilde S_3 \ketr{\rho}
  = 
  \frac{1}{\sqrt{d}} \sum_{\substack{a, a' \in \FF_2^{2n} \\ \symplecticp{a}{a'} = 0}} \frac{s_{a,a'}}{s_a s_{a'}} \braketr{O}{\sn_a} \braketr{O}{\sn_{a'}} \brar{\sn_{a+a'}}\bar\Lambda_{a,a'}\ketr{\rho}
\end{equation}
\end{proof}

In analogy to \cref{thm: Error bound}, we obtain the following bound on the deviation of the second moment from its value in the absence of noise.

\begin{prop} \label{prop: Variance bias}
In the setting of \cref{thm: Error bound}, assume that $|s_{a,a'}|/(s_a s_a') \leq C$ for all $a\neq a'$ with $\braketr{O}{\s_a}\neq 0$ and $\braketr{O}{\s_{a'}}\neq 0$.
Then, we have 
\begin{equation}
  |\EE(\estimator{o}^2) - \EE(\estimator{o}_\text{noise-free}^2)|
  \leq
  C \stnorm{O}^2 \max_{a,b\in\FF_2^{2n}} \, \dnorm{\id-\bar\Lambda_{a,b}} 
  \leq
  C \stnorm{O}^2 \max_{g\in G} \, \dnorm{\id-\Lambda(g)} \,,
\end{equation}
where $\estimator{o}_\text{noise-free}$ is the shadow estimator in the absence of any noise and $\bar\Lambda_{a,b}$ are suitably averaged noise channels.
\end{prop}

\begin{proof}
We have the following expression for $s_{a,a'}$:
\begin{equation}
 s_{a,a'}
 =
 \sqrt{d} \, \brar{\sn_a \otimes \sn_{a'}}S_3\ketr{\sn_{a+a'}}
 =
 \sqrt{d} \sum_{x\in\FF_2^n} \sum_{g\in G} p(g) \brar{\sn_a}\omega(g)^\dagger\ketr{E_x}\brar{\sn_{a'}}\omega(g)^\dagger\ketr{E_x} \brar{E_x}\omega(g)\ketr{\sn_{a+a'}} \,.
 \label{eq:sab}
\end{equation}
First, note that if $a=a'$, then $\sn_{a+a}=\sn_0=\1/\sqrt{d}$ and hence
\begin{equation}
 s_{a,a}
 = \sum_{x\in\FF_2^n} \sum_{g\in G} p(g) \brar{\sn_a}\omega(g)^\dagger\ketr{E_x}^2
 = \sum_{x\in\FF_2^n} \sum_{g\in G} p(g) \brar{\sn_a}\omega(g)^\dagger\ketr{E_x}\brar{E_x}\omega(g)\ketr{\sn_a}
 = \brar{\sn_a}S\ketr{\sn_a}
 = s_a\,,
 \label{eq:saa}
\end{equation}
using that all matrix coefficients are real.
By assumption, $|s_{a,a'}|/(s_a s_a') \leq C$ for all non-zero terms, and thus we find
\begin{align*}
  |\EE(\estimator{o}^2) - \EE(\estimator{o}_\text{noise-free}^2)|
  &= 
  |\brar{O \otimes O} S^{-1}\otimes S^{-1} (\tilde S_3 - S_3) \ketr{\rho}| \\
  &= 
  \frac{1}{\sqrt{d}} \Big| \sum_{\substack{a, a' \in \FF_2^{2n} \\ \symplecticp{a}{a'} = 0}} \frac{s_{a,a'}}{s_a s_{a'}} \braketr{O}{\sn_a}\braketr{O}{\sn_{a'}} \brar{\sn_{a+a'}}\bar\Lambda_{a,a'}-\id\ketr{\rho} \Big| \\
  &\leq
  \frac{C}{d^2} \sum_{a\neq a'}|\braketr{O}{\s_a}| |\braketr{O}{\s_{a'}}| \, \dnorm{\bar\Lambda_{a,a'}-\id} 
  +
  \frac{1}{d^2} \sum_{a} \frac{|\braketr{O}{\s_a}|^2}{s_a}  \, \Tr\left[(\bar\Lambda_{a,a'}-\id)(\rho)\right]\\
  &\leq
  C\,\stnorm{O}^2 \max_{a,a'} \, \dnorm{\bar\Lambda_{a,a'}-\id} \,,
\end{align*}
using that $\bar\Lambda_{a,a'}$ is trace preserving.
\end{proof}

An open question is whether we have $|s_{a,a'}|/(s_a s_a') = \LandauO(1)$ for general distributions $p$ on the Clifford group.
The best general upper bound we could find is $d$, but we think that this is too pessimistic for practically relevant cases.
For uniform sampling from a Clifford subgroup, we can -- in principle -- get an analytical handle on the $s_{a,a'}$ (and $s_a$) using Schur's lemma and information about the irreps of the subgroup. 
This can help to improve on \cref{prop: Variance bias}, as we illustrate in \cref{app:variance Cliffords} and \cref{app:variance local Cliffords} at the case of the local and global Clifford groups.

\subsection{The global Clifford group}
\label{app:variance Cliffords}

For the global Clifford group $\Cl_n$, the representation $\omega$ is composed of two irreducible representation and  can be decomposed as $\omega = \tau_0 \oplus \tau_1$. The noiseless second moment operator $S_3$ for uniform sampling from $\Cl_n$ then decomposes as (see also Ref.~\cite[{App.~C}]{Heinrich22GeneralGuarantees}):
\begin{align*}
S_3 = \frac{1}{|G|} \bigoplus_{i \in \FF_2}\bigoplus_{j \in \FF_2}\bigoplus_{k \in \FF_2} \sum_{g \in G} \tau_i(g)^{\dagger} \otimes \tau_j(g)^{\dagger} M_3 \tau_k(g)\, .
\end{align*}
By Schur's lemma, the operator
\begin{equation}
 \Pi_{ijk} = \frac{1}{|G|} \sum_{g \in G} \tau_i(g)^{\dagger} \otimes \tau_j(g)^{\dagger} (\cdot) \tau_k(g)
\end{equation}
is an orthogonal projector which is only non-zero if the irrep $\tau_k$ is contained in $\tau_i \otimes \tau_j$.
It is straightforward to see that $\Pi_{ijk}$ is thus zero for $(ijk) = (001), (100), (010)$.
More generally, $\rank \Pi_{ijk}$ is equal to the multiplicity of $\tau_k$ in $\tau_i \otimes \tau_j$.
Thus, $\Pi_{ijk}$ is rank one for $(ijk) = (000), (101), (011), (110)$ and rank two (one) for $(ijk) = (111)$ if $n>1$ ($n=1$) \cite[{Eq.~(97)}]{Heinrich22GeneralGuarantees}.
The rank one cases can be straightforwardly computed by finding a superoperator $I_{ijk}$ in the range of $\Pi_{ijk}$ and projecting $M_3$ onto $I_{ijk}$, i.e.
\begin{align*}
  \Pi_{ijk}(M_3) = \frac{\Tr(I_{ijk}^\dagger M_3)}{\Tr(I_{ijk} I_{ijk}^{\dagger})} I_{ijk}
\end{align*}
The rank two case was already computed in Ref.~\cite[{App.~C.1}]{Heinrich22GeneralGuarantees}.
Evaluating $S_3$ in this way results in the following Lemma.

\begin{lem}
For uniform sampling from the global Clifford group, the coefficients $s_{a,a'}/(s_a s_{a'})$ in the second moment are given by 
\begin{align}
\frac{s_{a,a'}}{s_a s_{a'}} = 
\begin{cases} 
1 & a = 0 \lor a' = 0 \\
d+1 & a \neq 0 \land a' \neq 0 \land a = a'\\
\frac{2(d+1)}{d+2}(-1)^{\beta(a,a')} & a\neq 0 \,\, \land \,\, a' \neq 0 \,\, \land \,\, a \neq a' \,\, \land \,\, \symplecticp{a}{a'} = 0 \\
0 & \text{else.}
\end{cases}
\end{align}
\end{lem}

\begin{proof}
We will first determine the superoperators $I_{ijk}$ and calculate each individual contribution to $S_3$. 
\begin{enumerate}[(i)]
\item We can choose $I_{000} = \kkbrar{\onen}{\onen}{\onen}$ with $\Tr(I^{\dagger}_{000}I_{000}) = 1$ and 
  \begin{align*}
    \Tr(I_{000}^\dagger M_3) 
    = \frac{1}{\sqrt{d}} \sum_{z,z' \in \FF_2^n} 
    \braketr{\onen}{\normalize Z_z} \braketr{\onen}{\normalize Z_{z'}} \braketr{\normalize Z_{z+z'}}{\onen} 
    = \frac{1}{\sqrt{d}} \, .
  \end{align*}
  The first term in $S_3$ is thus 
  \begin{align} \label{eq: I_000 term}
    \frac{\Tr(I_{000}^\dagger M_3)}{\Tr(I^{\dagger}_{000} I_{000})} I_{000} = \frac{1}{\sqrt{d}} \kkbrar{\onen}{\onen}{\onen}\, .
  \end{align}

\item 
  Next we look at $(ijk)=(011),(101)$.
  Note that $\sum_{a \neq 0} \ketbrar{\sn_a}{\sn_a}$ is the projector onto the traceless subspace and hence left invariant by $\tau_1^\dagger(g) (\cdot) \tau_1(g)$ for all $g\in G$.
  Consequently $I_{011} = \sum_{a \neq 0} \kkbrar{\onen}{\sn_a}{\sn_a}$,  $I_{101} = \sum_{a \neq 0} \kkbrar{\sn_a}{\onen}{\sn_a}$ are valid choices with $\Tr(I_{011}^{\dagger}I_{011}) = \Tr(I_{101}^{\dagger}I_{101}) = d^2-1$. 
  For the overlap with $M_3$ we find 
  \begin{align*}
  \Tr(I_{011}^\dagger M_3) 
  &= \frac{1}{\sqrt{d}}\sum_{z,z' \in \FF_2^n} \sum_{a \neq 0} \braketr{\onen}{\n Z_z}\braketr{\sn_a}{\n Z_{z'}} \braketr{\n Z_{z+z'}}{\sn_a}\\
  & = \frac{1}{\sqrt{d}}\sum_{z,z' \in \FF_2^n} \sum_{a \neq 0} \delta_{z,0}\delta_{z',a}\delta_{z+z',a} \\
  & = \frac{1}{\sqrt{d}}\sum_{z' \neq 0} 1\\
  & = \frac{1}{\sqrt{d}} (d-1)\, .
  \end{align*}
  The same result can be obtained for $\Tr(I_{101}^\dagger M_3)$, and we get the contributions 
  \begin{align} \label{eq: I_101 and I_011 term}
    \frac{\Tr(I_{011}^\dagger M_3)}{\Tr(I^{\dagger}_{011} I_{011})} I_{011}
    &= \frac{1}{\sqrt{d}(d+1)} \sum_{a \neq 0} \kkbrar{\onen}{\sn_a}{\sn_a}\, ,&
    \frac{\Tr(I_{101}^\dagger M_3)}{\Tr(I^{\dagger}_{101} I_{101})} I_{101} 
    &= \frac{1}{\sqrt{d}(d+1)} \sum_{a \neq 0} \kkbrar{\sn_a}{\onen}{\sn_a} \, .
  \end{align}
\item 
  The $(110)$ case is treated in \cite[{App.~C.1a}]{Heinrich22GeneralGuarantees} using
  \begin{align*}
  I_{110} = \left[F - \ketr{\onen \otimes \onen}\right] \brar{\onen},
  \end{align*}
   where $F$ is the flip operator on the first two tensor factors. 
   To compute $\Tr(I_{110}^\dagger M_3)$, we first remember the property of the flip operator that $\Tr(F \sn_a \otimes \sn_b) = \braketr{\sn_a}{\sn_b} = \delta_{a,b}$, which leads us to 
  \begin{align*}
      \Tr(I_{110}^\dagger\ketr{\sn_a \otimes \sn_b}\brar{\sn_c}) 
      = \big[\Tr(F \sn_a \otimes \sn_b) - \braketr{\onen}{\sn_a} \braketr{\onen}{\sn_b} \big] \braketr{\onen}{\sn_c} 
      = (\delta_{a,b} - \delta_{a,0}\delta_{b,0})\delta_{c,0} \, .
  \end{align*}
  Therefore, we can write $I_{110}$ in the Pauli basis as $I_{110} = \sum_{a\neq 0} \kkbrar{\sn_a}{\sn_a}{\onen}$, from where one can deduce the normalization $\Tr(I_{110}^{\dagger}I_{110}) = d^2-1$. For the overlap with $M_3$ we have  
  \begin{align*}
    \Tr(I_{110}^\dagger M_3) 
    &= \frac{1}{\sqrt{d}} \sum_{z,z' \in \FF_2^n} \bigl(\delta_{z,z'} - \delta_{z,0}\delta_{z',0} \bigr) \delta_{z+z',0}
    = \frac{d-1}{\sqrt{d}} \, . 
  \end{align*}
  The contribution to $S_3$ is then given as
  \begin{align}\label{eq: I_110 term}
    \frac{\Tr(I_{110}^\dagger M_3)}{\Tr(I^{\dagger}_{110} I_{110})} I_{110} = \frac{1}{\sqrt{d}(d+1)}\sum_{a \neq 0} \kkbrar{\sn_a}{\sn_a}{\onen}\, .
  \end{align}
  \item 
    Lastly we take a look at $I_{111}$, which was also determined in \cite[{App.~C.1b}]{Heinrich22GeneralGuarantees}, where it was shown that 
    \begin{align}
      \Tr(I_{111}M_3) = \frac{\Tr(I_\mathrm{ad}^{(1)}M_3)}{d^3 (d^2-1)(d^2-4)} (I_\mathrm{ad}^{(1)} + I_\mathrm{ad}^{(2)})\, 
    \end{align}
    with 
    \begin{align}
    I_\mathrm{ad}^{(1)} &= \sum_{\substack{a\neq 0, b \neq 0 \\ a \neq b }} \kkbrar{\s_a}{\s_b}{\s_a \s_b} \,, &
    I_\mathrm{ad}^{(1)} &= \sum_{\substack{a\neq 0, b \neq 0 \\ a \neq b }} \kkbrar{\s_a}{\s_b}{\s_b \s_a} \, .
    \end{align}
    Using $\s_a \s_b = (-1)^{\symplecticp{a}{b}} \s_b \s_a$ and the above definition we find that 
    \begin{align}
      I_\mathrm{ad}^{(1)} + I_\mathrm{ad}^{(2)} = 2\sum_{\substack{a\neq 0, b \neq 0 \\ a \neq b, \symplecticp{a}{b} = 0 }} \kkbrar{\s_a}{\s_b}{\s_a\s_b}\, ,
    \end{align}
    as well as 
    \begin{align*}
      \Tr(I_\mathrm{ad}^{(1)\dagger} M_3) 
      = d \sum_{z,z' \in \FF_2^n} \sum_{\substack{a\neq 0, b \neq 0 \\ a \neq b, \symplecticp{a}{b} = 0 }} 
        (-1)^{\beta(a,b)}\delta_{z,a} \delta_{z',b}\delta_{z+z',a+b} 
      = d \sum_{\substack{z \neq 0, z' \neq 0 \\ z \neq z'}} 1 
      = d(d-1)(d-2)\, .
    \end{align*}
    The contribution to $S_3$ is thus
    \begin{align} 
      \frac{\Tr(I_{111}^\dagger M_3)}{\Tr(I^{\dagger}_{111} I_{111})} I_{111} 
      &= \frac{2}{d^2(d+1)(d+2)} \sum_{\substack{a\neq 0, a' \neq 0 \\ a \neq a', \symplecticp{a}{a'} 
      = 0 }} \kkbrar{\s_a}{\s_{a'}}{\s_a \s_{a'}} \\
      &= \frac{2}{\sqrt{d}(d+1)(d+2)} \sum_{\substack{a\neq 0, a' \neq 0 \\ a \neq a', \symplecticp{a}{a'} 
      = 0 }} (-1)^{\beta(a,a')}\kkbrar{\sn_a}{\sn_{a'}}{\sn_{a+a'}}\, . \label{eq: I_111 term}
    \end{align}
\end{enumerate}

By comparing the prefactors of different terms from $S_3$ given in \cref{eq: I_000 term}, \cref{eq: I_101 and I_011 term}, \cref{eq: I_110 term} and \cref{eq: I_111 term} with \cref{eq:S3 definition}, we can read off the coefficients $s_{a,a'}$:
\begin{align}
s_{a,a'} = 
\begin{cases} 
1 & a = 0 \lor a' = 0 \\
\frac{1}{d+1} & a \neq 0 \land a' \neq 0 \land a = a'\\
\frac{2}{(d+1)(d+2)}(-1)^{\beta(a,a')} & a\neq 0 \,\, \land \,\, a' \neq 0 \,\, \land \,\, a \neq a' \,\, \land \,\, \symplecticp{a}{a'} = 0 \\
0 & \text{else.}
\end{cases}
\end{align}
 For the second moment $\EE(\estimator{o})$ we need the fractions $\frac{s_{a,a'}}{s_a s_a'}$, which using $s_a = (d+1)^{-1}$ for $a \neq 0$ and $s_0 = 1$ are found to be
\begin{align} \label{eq-app: s-coefficient cases}
\frac{s_{a,a'}}{s_a s_{a'}} = 
\begin{cases} 
1 & a = 0 \lor a' = 0 \\
d+1 & a \neq 0 \land a' \neq 0 \land a = a'\\
\frac{2(d+1)}{d+2}(-1)^{\beta(a,a')} & a\neq 0 \,\, \land \,\, a' \neq 0 \,\, \land \,\, a \neq a' \,\, \land \,\, \symplecticp{a}{a'} = 0 \\
0 & \text{else.}
\end{cases}
\end{align}
\end{proof}

\subsection{The local Clifford group}
\label{app:variance local Cliffords}

For the local Clifford group, we begin by arguing that for uniform sampling and without the presence of noise, the frame operator factorizes. First note that the Z-basis measurement operator factorizes as
\begin{align}
 M_3 = \frac{1}{\sqrt{d}} \left(\sum_{z,z' \in \FF_2} \kkbrar{\n Z_z}{\n Z_{z'}}{\n Z_{z+z'}}\right)^{\otimes n} \, .
\end{align}
Consequently, we find:
\begin{align}
S_3 
&= \frac{d^{-1/2}}{\abs{\Cl_1^{\otimes n}}} \sum_{g \in \Cl_1^{\otimes n}} \omega^{\dagger}(g)^{\otimes 2} M_3 \omega(g) \\
&= \frac{1}{\sqrt{d}} \left(\frac{1}{\abs{\Cl_1}}\sum_{g \in \Cl_1}\sum_{z,z' \in \FF_2}  \omega^{\dagger}(g)^{\otimes 2} \kkbrar{\n Z_z}{\n Z_{z'}}{\n Z_{z+z'}} \,  \omega(g)\right)^{\otimes n} \\
&= \frac{1}{\sqrt{d}} \left(\sum_{a,a' \in \FF_2^2} s_{a,a'} \kkbrar{\sn_a}{\sn_{a'}}{\sn_{a+a'}}\right)^{\otimes n}\, ,
\end{align}
where we used \cref{eq:S3 definition} for the case $n=1$ in the last line.
Since $S$ also factorizes for the local Clifford group, one can easily verify that $S^{-1} \otimes S^{-1} S_3 = 
\frac{1}{\sqrt{d}} \left(\sum_{a,a' \in \FF_2^2} \frac{s_{a,a'}}{s_a s_{a'}} \kkbrar{\sn_a}{\sn_{a'}}{\sn_{a+a'}}\right)^{\otimes n}$. The coefficients were already determined in \cref{eq-app: s-coefficient cases} and for $n=1$ we can simplify them to get (the third case cannot occur):
\begin{align} \label{eq-app: local s-coefficient cases}
\frac{s_{a,a'}}{s_a s_a'} \overset{n = 1}{=} 
\begin{cases} 
1 & a = 0 \lor a' = 0 \\
3 & a \neq 0 \land a' \neq 0 \land a = a'\\
0 & \text{else.}
\end{cases}
\end{align}

For $n>1$ qubits, recall that we label Pauli operators as $\sigma_a = \sigma_{a_1} \otimes \dots \otimes \sigma_{a_n}$ where $a_i \in \FF_2^2$.
The coefficients are then found as products of the above single qubit coefficients:
\begin{align}\label{eq-app: k-local s-coefficients}
\frac{s_{a,a'}}{s_a s_a'}  =  \prod_{i = 1}^n \frac{s_{a_i,a'_i}}{s_{a_i} s_{a'_i}} =
\begin{cases} 
0 & \exists i \in \supp(a) \cap \supp(a'):\, a_i \neq a'_i \\
3^{|\supp(a) \cap \supp(a')|} & \text{else.}
\end{cases}
\end{align}

\subsection{Variance bound}
\thmVariance*
\begin{proof}
The variance is given by $\Var[\hat o] = \EE [\hat o - \EE [\hat o]]^2$ and $O = O_0 + \Tr(O) \ii/d$ as before.
We have $\sandwichr{\ii}{S^{-1}\omega^{\dagger}(g)}{E_x} = 1$ and $\sandwichr{\ii}{S^{-1}\tilde S}{\rho} = \braketr{\ii}{\rho} = 1$ since the involved superoperators and the noise is trace preserving.
We then find:
\begin{align}
  \hat o(g,x) - \EE [\hat o] = \sandwichr{O}{S^{-1}\omega^{\dagger}(g)}{E_x} - \sandwichr{O}{S^{-1}\tilde S}{\rho} = 
  \sandwichr{O_0}{S^{-1}\omega^{\dagger}(g)}{E_x} - \sandwichr{O_0}{S^{-1}\tilde S}{\rho} = \hat o_0(g,x) - \EE [\hat o_0]
\end{align}
and the variance only depends on the traceless part $\Var(\hat o) = \Var(\hat o_0)$.
We will now bound the variance for the global Clifford protocol in terms of the second moment: $\Var[\hat o_0] = \EE[\hat o_0^2] - (\EE [\hat o_0])^2 \leq \EE[\hat o_0^2]$.
To rewrite the second moment given in \cref{eq: noisy second moment}, we first note that $\bar \Lambda_{a,a'} = \bar \Lambda_{a',a}$ which can be seen from its definition and from $s_{a,a'} = s_{a',a}$. 
The second moment for a traceless observable $O_0$ can then be written as 
\begin{align}
  \EE(\estimator{o}_0^2)
  &= \frac{1}{\sqrt{d}} \sum_{\substack{a \neq 0, a' \neq 0: \\ \symplecticp{a}{a'} = 0}} 
  \frac{s_{a,a'}}{s_a s_{a'}} \braketr{O_0}{\sn_a} \braketr{O_0}{\sn_{a'}} \brar{\sn_{a+a'}}\bar\Lambda_{a,a'}\ketr{\rho} \\
  &= \frac{1}{\sqrt{d}} \frac{2(d+1)}{d+2} \sum_{\substack{a\neq 0, a' \neq 0 \\ a \neq a', \symplecticp{a}{a'} = 0 }} (-1)^{\beta(a,a')} \braketr{O_0}{\sn_a} \braketr{O_0}{\sn_{a'}} \brar{\sn_{a+a'}}\bar\Lambda_{a,a'}\ketr{\rho} + \frac{d+1}{\sqrt{d}} \sum_{a\neq 0}\braketr{O_0}{\sn_a}^2 \brar{\onen}\bar\Lambda_{a,a}\ketr{\rho}\, . \\
\end{align}
We can bound the sums by using $\abs{\sandwichr{\sn_{a+a'}}{\bar\Lambda_{a,a'}}{\rho}} \leq \frac{1}{\sqrt{d}}$ to get

\begin{align}
\EE(\estimator{o}_0^2) &\leq \frac{2(d+1)}{d(d+2)}\sum_{\substack{a\neq 0, a' \neq 0 \\ a \neq a', \symplecticp{a}{a'} = 0 }} \abs{\braketr{O_0}{\sn_a}\braketr{O_0}{\sn_{a'}}} 
  + \frac{d+1}{d} \twonorm{O_0}^2 \\
  &\leq \frac{2(d+1)}{d(d+2)}\sum_{a\neq 0} \abs{\braketr{O_0}{\sn_a}} \sum_{a' \neq 0} \abs{\braketr{O_0}{\sn_{a'}}}  + \frac{d+1}{d} \twonorm{O_0}^2 \\
  &= \frac{2(d+1)}{d+2} \stnorm{O_0}^2  + \frac{d+1}{d} \twonorm{O_0}^2 \, .
\end{align} 

To bound the second moment for the local Clifford group, we start at the general form for Clifford protocols given in \cref{eq: noisy second moment} and use again that $\abs{\sandwichr{\sn_{a+a'}}{\bar\Lambda_{a,a'}}{\rho}} \leq \frac{1}{\sqrt{d}}$. 
This leads us to 
\begin{align}
    \EE(\estimator{o}^2) 
    &= \frac{1}{\sqrt{d}} \sum_{\substack{a, a' \in \FF_2^{2n}: \\ \symplecticp{a}{a'} = 0}} 
      \frac{s_{a,a'}}{s_a s_{a'}} \braketr{O}{\sn_a} \braketr{O}{\sn_{a'}} \sandwichr{\sn_{a+a'}}{\bar\Lambda_{a,a'}}{\rho}  
    \leq \frac{1}{d} \sum_{\substack{a, a' \in \FF_2^{2n}: \\ \symplecticp{a}{a'} = 0}} 
      \frac{s_{a,a'}}{s_a s_{a'}} \abs{\braketr{O}{\sn_a} \braketr{O}{\sn_{a'}}} \, . \label{eq-app: local clifford second moment bound}
\end{align}

We first consider $k$-local Pauli observables $P = \sigma_{b_1}\otimes \cdots \otimes \sigma_{b_n}$, where $b_i = 0$ on $n-k$ sites.
Then, \cref{eq-app: local clifford second moment bound} factorizes and using \cref{eq-app: local s-coefficient cases} we find that 
\begin{align} \label{eq-app: local Clifford local Pauli bound}
    \EE(\estimator{o}_P^2) 
    \leq \frac{1}{d} \prod_{i = 1}^n  \sum_{\substack{a_i, a'_i \in \FF_2^2: \\ \symplecticp{a_i}{a'_i} = 0}} 
      \frac{s_{a_i,a'_i}}{s_{a_i} s_{a'_i}} \abs{ \braketr{\s_{b_i}}{\sn_{a_i}} \braketr{\s_{b_i}}{\sn_{a_i'}} } 
    = \frac{1}{d} \prod_{i = 1}^n  \sum_{\substack{a_i, a'_i \in \FF_2^2: \\ \symplecticp{a_i}{a'_i} = 0}} 
      \frac{s_{a_i,a'_i}}{s_{a_i} s_{a'_i}} d \, \delta_{a_i,b_i} \delta_{a_i',b_i}
    = 3^k\,.
\end{align}
The proof for $k$-local observables $O = O_\mathrm{loc}\otimes \ii^{\otimes (n-k)}$ is almost identical to the proof without noise given in \citet{Huang2020Predicting}, and we repeat the relevant steps here. Akin to \cref{eq-app: local Clifford local Pauli bound} we find that
\begin{align} 
    \EE(\estimator{o}^2) &\leq \frac{1}{2^k} \sum_{\substack{a, a' \in \FF_2^{2k}: \\ \symplecticp{a}{a'} = 0}} \frac{s_{a,a'}}{s_{a} s_{a'}} \abs{ \braketr{O_\mathrm{loc}}{\sn_a} \braketr{O_\mathrm{loc}}{\sn_{a'}} }\, .
    \label{eq:local Clifford local observable}
\end{align}
In the following, we drop the condition $\symplecticp{a}{a'} = 0$ in the sum since we have $s_{a,a'} = 0$ for $\symplecticp{a}{a'} \neq 0$.
Let $\FF_2^{2*} = \FF_2^2\setminus\{0\}$ label the single-qubit non-identity Paulis and define a partial order on $\FF_2^{2k}$ such that $a \leq b$ iff for all $i \in [k]$ either $a_i = b_i$ or $a_i = 0$ holds. 
By \cref{eq-app: k-local s-coefficients}, pairs $a, a' \in \FF_2^{2k}$ which do not coincide on their common support do not contribute to the sum in \cref{eq:local Clifford local observable}.
For the remaining pairs, we can always find a $b \in (\FF_2^{2*})^k$ such that $a\leq b$ and $a'\leq b$. 
We can thus replace the sum in \cref{eq:local Clifford local observable} as $\sum_{a, a' \in \FF_2^{2k}} \rightarrow \sum_{b \in (\FF_2^{2*})^k}\sum_{a \leq b}\sum_{a' \leq b}$ if we take care of potential over counting. 
To this end, note that we are free to choose $b_i$ for every index $i$ where $a_i = a'_i = 0$ and there are in total $k - \abs{\supp(a) \cup \supp(a')} = k - \abs{\supp(a)} - \abs{\supp(a')} + \abs{\supp(a)\cap \supp(a')}$ of such indices.
Using $\frac{s_{a,a'}}{s_{a} s_{a'}} = 3^{\abs{\supp(a)\cap \supp(a')}}$, we find:
\begin{align}
\frac{1}{2^k} \sum_{a, a' \in \FF_2^{2k}} \frac{s_{a,a'}}{s_{a} s_{a'}} \abs{ \braketr{O_\mathrm{loc}}{\sn_a} \braketr{O_\mathrm{loc}}{\sn_{a'}} }
 &= \frac{1}{2^k} \sum_{b \in (\FF_2^{2*})^k}\sum_{\substack{a \leq b \\ a' \leq b}} 
 \frac{3^{\abs{\supp(a)} + \abs{\supp(a')}}}{3^{k}} \abs{ \braketr{O_\mathrm{loc}}{\sn_a} \braketr{O_\mathrm{loc}}{\sn_{a'}} } \\
 &= \frac{1}{2^k 3^k} \sum_{b \in (\FF_2^{2*})^k} \left(\sum_{a \leq b} 3^{\abs{\supp(a)}} \abs{\braketr{O_\mathrm{loc}}{\sn_{a}}}\right)^2
\, .
\end{align}
With $\sum_{a \leq b} 3^{\abs{\supp(a)}} = 4^k$ for all $b \in (\FF_2^{2*})^k$ \cite{Huang2020Predicting} and the Cauchy-Schwarz inequality, this can be simplified as follows:
\begin{align} 
\frac{1}{2^k 3^k} \sum_{b \in (\FF_2^{2*})^k} \left(\sum_{a \leq b} 3^{\abs{\supp(a)}}  \abs{\braketr{O_\mathrm{loc}}{\sn_{a}}}\right)^2 
&\leq \frac{1}{2^k 3^k} \sum_{b \in (\FF_2^{2*})^k} \sum_{a \leq b} 3^{\abs{\supp(a)}}  \sum_{a' \leq b} 3^{\abs{\supp(a')}} \braketr{O_\mathrm{loc}}{\sn_{a'}}^2 \\
&= \frac{4^k}{2^k} \sum_{b \in (\FF_2^{2*})^k} \sum_{a' \leq b} \frac{3^{\abs{\supp(a')}}}{3^k}\, \braketr{O_\mathrm{loc}}{\sn_{a'}}^2 
\, .
\end{align}
Note that in the last line, the number of times each $a' \in \FF_2^{2k}$ appears in the double sum is given by $3^{k-\abs{\supp(a')}}$ and therefore $\sum_{b \in (\FF_2^{2*})^k}\sum_{a'\leq b} \frac{3^{\abs{\supp(a')}}}{3^k} f(a') = \sum_{a' \in \FF_4^k} f(a')$ for arbitrary summands $f(a')$. 
This leads us to the final result 
\begin{align}
\frac{4^k}{2^k} \sum_{b \in (\FF_2^{2*})^k} \sum_{a' \leq b} \frac{3^{\abs{\supp(a')}}}{3^k}\, \braketr{O_\mathrm{loc}}{\sn_{a'}}^2 
= \frac{4^k}{2^k} \sum_{a' \in \FF_2^{2k}} \braketr{O_\mathrm{loc}}{\sn_{a'}}^2 
= 2^k \twonorm{O_\mathrm{loc}}^2 
\leq 4^k \snorm{O_\mathrm{loc}}^2
\, .
\end{align}

\end{proof}

\section{Robust shadow estimation under gate-dependent noise}\label{app: robust shadows}
A prominent noise mitigation technique for shadow estimation are the \emph{robust classical shadows} developed by \citet{Chen21RobustShadowEstimation}. 
Robust classical shadows rely on the assumption of gate-independent left-noise, 
i.e.\ 
\begin{equation}\label{eq:gate-independent_noise}
  \phi(g) = \Lambda \omega(g)
\end{equation}
for a $g$-independent channel $\Lambda$. 
In this case, the frame operator can be written as $\tilde{\mc S} = \sum_{\lambda \in \mathrm{Irr}(G)} f_{\lambda} \Pi_{\lambda}$, where $\mathrm{Irr}(G)$ is the set of irreducible representations of the group G and $\Pi_{\lambda}$ are projectors onto invariant subspaces. 
For the global Clifford group and trace preserving noise, this reduces to $\tilde S = \ketbrar{\onen}{\onen} + f \sum_{a \neq 0} \ketbrar{\sn_a}{\sn_a}$, meaning that a single parameter $f$ needs to be estimated in order to have full knowledge of $\tilde S$, which allows for the mitigation with $\tilde S^{-1}$ in post-processing. In \citet{Chen21RobustShadowEstimation}, $f$ is determined by the median of means of the single shot estimator 
\begin{equation}
  \hat f(g,x) \coloneqq (d \brar{E_x}\omega(g)\ketr{E_0} - 1)/(d-1)\, 
\end{equation}
where for each round $\ket{0}$ state is prepared, a random operation $g \in G$ is applied and the result $x$ is stored. 
It is then shown that the expectation value for gate-independent left noise satisfies $\EE[\hat f(g,x)] = f$. 

We will now show how gate-independent left noise can be written in our gate-dependent right noise model and the effect it has on the noisy frame operator. This provides an understanding as to how left- and right noise interrelate and shows consistency of our formalism with previous works on gate independent noise \cite{Chen21RobustShadowEstimation,KohGrewal:2022:ClassicalShadows}.
The assumption \eqref{eq:gate-independent_noise} translates to $\phi(g) = \omega(g) \omega^{\dagger}(g) \Lambda \omega(g) = \omega(g) \Lambda(g)$ for $\Lambda(g) =  \omega^{\dagger}(g) \Lambda \omega(g)$. 
The average noise channels $\bar \Lambda_a$ of \Cref{lem: Noise averaging} for uniform sampling from the global Clifford group are given by 
\begin{align}\label{eq: gate independent noise average}
\bar \Lambda_a  = \frac{(d+1)}{|\Cl_n|}\sum_{\vec z \in \FF_2^n} \sum_{g \in \Xi^{-1}_{z}(a)} \Lambda(g) = \frac{(d+1)}{|\Cl_n|}\sum_{\vec z \in \FF_2^n} \sum_{g \in \Xi^{-1}_{z}(a)} \omega^{\dagger}(g) \Lambda \omega(g).
\end{align}
We know that for every $g\in \Xi^{-1}_{z}(a)$, $hg \in \Xi^{-1}_{z}(a)$ with $h \in \HW_n$, since
\begin{align}
\omega^{\dagger}(hg) \ketbrar{\normalize Z_z}{\normalize Z_z}\omega(hg) = \omega^{\dagger}(g) \omega^{\dagger}(h) \ketbrar{\normalize Z_z}{\normalize Z_z}\omega(h) \omega(g) = \omega^{\dagger}(g) \ketbrar{\normalize Z_z}{\normalize Z_z}\omega(g).
\end{align}
This means that the average in \cref{eq: gate independent noise average} contains an average $\frac{1}{|\HW_n|} \sum_{h \in \HW_n}\omega^{\dagger}(g) \Lambda \omega(g)$, which is commonly referred to as a Pauli twirl and the result is a Pauli channel (see e.g.\ \cite{Flammia2019EfficientEstimation}). In the following we will restrict ourselves to $a \neq 0$ where we will need the inverse relation 
\begin{align}
\frac{(d+1)}{|\Cl_n|}\sum_{\vec z \in \FF_2^n} \sum_{g \in \Xi^{-1}_{z}(a)} \omega(g) \ketbrar{\sn_a}{\sn_a}\omega(g)^{\dagger} 
&= \frac{(d+1)}{|\Cl_n|}\sum_{\vec z \in \FF_2^n} \abs{\Xi^{-1}_{z}(a)} \ketbrar{\normalize Z_z}{\normalize Z_z}
= \frac{1}{d-1} \sum_{\vec z \in \FF_2^n} \ketbrar{\normalize Z_z}{\normalize Z_z}\,,
\end{align}
which can be verified using \Cref{prop: Global Clifford average structure}.
The diagonal entries of $\bar \Lambda_a$ for $a \neq 0$ can then be determined by
\begin{align}
  \frac{(d+1)}{|\Cl_n|}\sum_{\vec z \in \FF_2^n} \sum_{g \in \Xi^{-1}_{z}(a)} \brar{\sn_a}\omega(g)^{\dagger}\Lambda \omega(g) \ketr{\sn_a} 
  &=  \frac{(d+1)}{|\Cl_n|}\sum_{\vec z \in \FF_2^n} \sum_{g \in \Xi^{-1}_{z}(a)} \Tr\left[\Lambda \omega(g) \ketbrar{\sn_a}{\sn_a}\omega(g)^{\dagger}\right] \\
  &= \frac{1}{d-1} \Tr\left[ \Lambda \sum_{\vec z \in \FF_2^n\backslash \vec 0} \ketbrar{\normalize Z_z}{\normalize Z_z} \right] \, .
\end{align}
In summary, the average channels $\bar \Lambda_a$ are diagonal since they are Pauli channels, and they all share the same diagonal elements $\brar{\sn_a}\bar \Lambda_a \ketr{\sn_a}$, which are independent of $a$. Note that $\brar{\onen} \Lambda \ketr{\onen} = 1$ for trace preserving or unital noise and let $f =  \Tr\left[ \Lambda \sum_{\vec z \in \FF_2^n\backslash \vec 0} \ketbrar{\normalize Z_z}{\normalize Z_z} \right]/(d^2-1) =  (\Tr[\Lambda M]-1)/(d^2-1)$. Then the noisy frame operator is given by 
\begin{align}
\tilde S = \ketbrar{\onen}{\onen} + \frac{1}{d+1} \sum_{a \neq 0} \ketbrar{\sn_a}{\sn_a}\bar \Lambda_a  
= \ketbrar{\onen}{\onen} + f \sum_{a \neq 0} \ketbrar{\sn_a}{\sn_a}\, 
\end{align}
which is the result of \cite{Chen21RobustShadowEstimation} for the global Clifford group. \\

We will now turn our attention back to gate-dependent noise with the following lemma. First, we define $\Zsf^n \coloneqq \{ 00, 01 \}^n$, the index set for diagonal Paulis.

\begin{lem}\label{lem: Robust shadow parameter}
The mitigation parameter in the robust shadow estimation protocol for uniform sampling over the global Clifford group under gate-dependent Pauli noise is given by 
\begin{align} \label{eq-app: RS mitigation parameter}
\EE_{g,x} [\hat f(g,x)] = \frac{1}{d+1} \EE_{a \in \Zsf^n\backslash 0} \bar \lambda_a\, .
\end{align}
\end{lem}
\begin{proof}
We need to compute the expectation value of the single shot estimator $\hat f(g,x) = \frac{d \cdot \brar{E_x}\omega(g)\ketr{E_0} - 1}{d-1}$. 
In the first step we rewrite the expectation value of $\brar{E_x}\omega(g)\ketr{E_0}
 $.
\begin{align*}
\EE_{g,x} \brar{E_x}\omega(g)\ketr{E_0} 
&= \frac{1}{|\Cl_n|} \sum_{g \in \Cl_n} \sum_{x \in \FF_2^n} \brar{E_x}\omega(g)\ketr{E_0} \brar{E_x}\omega(g)\Lambda(g)\ketr{E_0} \\
&= \frac{1}{|\Cl_n|} \sum_{g \in \Cl_n} \sum_{x \in \FF_2^n} \brar{E_0}\omega(g)^{\dagger} \ketr{E_x}\brar{E_x}\omega(g)\Lambda(g)\ketr{E_0}\\
& = \brar{E_0} \tilde S \ketr{E_0}\, .
\end{align*}
By noting that $\ketr{E_0} = \left(\frac{1}{\sqrt{2}} (\ketr{\onen} + \ketr{\n Z})\right)^{\otimes n} = \frac{1}{\sqrt{d}}\sum_{z \in \FF_2^n} \ketr{\n Z(z)}$ we have $\brar{E_0} \tilde S \ketr{E_0} = \frac{1}{d}\sum_{z,z' \in \FF_2^n} \brar{\n Z(z)}\tilde S \ketr{\n Z(z')}$. 
With the use of \Cref{lem: Noise averaging} and $s_a = 1/(d+1)$ for $a \neq 0$ and $s_0 = 1$, we find:
\begin{align*}
\EE_{g,x} \brar{E_x}\omega(g)\ketr{E_0}  &= \frac{1}{d} + \frac{1}{d(d+1)} \sum_{a \in \Zsf^n\backslash 0} \bar \lambda_a\, .
\end{align*}
It follows that 
\begin{equation}
 \EE_{g,x} \hat f(g,x) = \frac{d \EE_{g,x} \brar{E_x}\omega(g)\ketr{E_0}  - 1}{d-1} = \frac{1}{(d+1)(d-1)}\sum_{a \in \Zsf^n\backslash 0} \bar \lambda_{a} = \frac{1}{d+1} \EE_{a \in \Zsf^n\backslash 0} \bar \lambda_a
\end{equation}
\end{proof}
Since \cref{eq-app: RS mitigation parameter} shows that for Pauli noise $\EE[\hat f(g,x)]$ always contains the factor $1/(d+1)$, we define the mitigation parameter as $\hat f_m \coloneqq (d+1) \EE[\hat f(g,x)]$.
The final estimate of the robust shadow estimation procedure is then given by
\begin{align}\label{eq: Robust shadow estimator}
\EE[ \estimator{o}_{RS}] \coloneqq \brar{O}\Big(\ketbrar{\onen}{\onen} + \frac{d+1}{\hat f_m} \sum_{a \neq 0} \ketbrar{\sn_a}{\sn_a}\Big) \tilde S \ketr{\rho}\,.
\end{align}

We will now show that whether the robust shadow estimation strategy succeeds for gate-dependent noise is highly dependent on the initial state, the observable and the noise present in the experiment.
As can be seen via a simple example, the error can actually be dramatically increased if the noise is not of left gate-independent form. 

\begin{restatable}{prop}{RSerror}\label{prop: Robust shadow exponential error}
  Under gate-dependent local noise $\phi_{\epsilon}(g) = (1-\epsilon) \omega(g) + \epsilon\, \omega(g)\Lambda(g)$, using the robust shadow estimate $\hat{o}_\mathrm{RS}$ with the global Clifford group can introduce a bias $\bias{o_\mathrm{RS}}{O} \geq \abs{\av{O_0}(\frac{1}{2}(1 + \epsilon)^n-1 )}$.
\end{restatable}
\begin{proof}
Consider the local bit flip channel $\Lambda(g) = \mc X$ for all $g \in \Cl_n$ and let $\Lambda_{\epsilon}$ be the bit flip channel with error probability $\epsilon$, i.e. $\Lambda_{\epsilon} = (1-\epsilon) \id + \epsilon \Lambda$. Its process matrix is given by $\Lambda_{\epsilon} = \ketbrar{\onen}{\onen} + \ketbrar{\n X}{\n X} + (1-2\epsilon)(\ketbrar{\n Y}{\n Y} + \ketbrar{\normalize Z}{\normalize Z})$.  If this Pauli channel acts on all qubits before each Clifford gate, then the global implementation map is $\phi_{\epsilon}(g) = \omega(g) \Lambda_{\epsilon}^{\otimes n}$.
Since this right noise channel is gate-independent, per \Cref{lem: Noise averaging} we get $\tilde{S} = (\ketbrar{\onen}{\onen} + \frac{1}{d+1} \sum_{a \neq 0}\ketbrar{\sn_a}{\sn_a}) \cdot \Lambda_{\epsilon}^{\otimes n}$ and $S^{-1}\tilde S = \Lambda_{\epsilon}^{\otimes n}$.
The robust shadow mitigation parameter $\EE_{g,x} \hat f(g,x)$ is per \Cref{lem: Robust shadow parameter} given by 
\begin{align}
\EE_{g,x} \hat f(g,x) &= \frac{1}{(d-1)(d+1)}\sum_{a \in \Zsf \backslash 0} \brar{\sn_a}\Lambda_{\epsilon}^{\otimes n}\ketr{\sn_a} \\
&= \frac{1}{(d-1)(d+1)} \left(\sum_{\vec z \in \FF_2^n} 1^{n-|\vec z|} (1-2\epsilon)^{|\vec z|} - 1\right)\\
&= \frac{1}{(d-1)(d+1)} \left(\sum_{i = 1}^n \binom{n}{i} 1^{n-i} (1-2\epsilon)^{i} - 1\right)\\
&= \frac{1}{(d+1)(d-1)}\left(d(1-\epsilon)^n -1\right).
\end{align}
In the robust shadow estimation procedure, the inverse frame operator is then given by 
\begin{equation}
S_{\RS}^{-1} = \ketbrar{\onen}{\onen} + (d+1) \frac{d-1}{d(1-\epsilon)^n -1} \sum_{a \neq 0}\ketbrar{\sn_a}{\sn_a},
\end{equation}
and the expected outcome is
\begin{equation}
\EE [\estimator{o}_{\RS}] =
\brar{O}S_{\RS}^{-1} \tilde S \ketr{\rho} = 
\frac{\Tr(O)}{d} + \frac{d-1}{d(1-\epsilon)^n -1} \sum_{a \neq 0}\braketr{O}{\sn_a}\brar{\sn_a} \Lambda_{\epsilon}^{\otimes n}\ketr{\rho} \, .
\end{equation}
Consequently, errors on Y-type and Z-type Pauli observables are partially mitigated while previously absent errors on X-Pauli observables are introduced. In particular, for the stabilizer state $O = (\ketbra{+}{+})^{\otimes n} = \frac{1}{\sqrt{d}}\sum_{x} \ketr{\normalize X^{x}}$, we find that $\EE [\estimator{o}] =  \bra{+}\rho\ket{+}$, meaning the standard shadow estimate is unbiased. 
Moreoever, we have $\brar{O}\Lambda_{\epsilon}^{\otimes n} = \brar{O}$ and thus
\begin{equation}
\EE [\estimator{o}_{\RS}] - \frac{\Tr(O)}{d} 
=
\frac{d-1}{d(1-\epsilon)^n -1} \left( \sandwichr{O}{\Lambda_{\epsilon}^{\otimes n}}{\rho} - \frac{\Tr(O)}{d} \right)
=
\frac{d-1}{d(1-\epsilon)^n -1} \av{O_0}\,,
\end{equation}
using the unitality of $\Lambda_{\epsilon}^{\otimes n}$ and $\av{O_0} = \braketr{O_0}{\rho}$.
Hence, 
\begin{align}
\EE [\estimator{o}_{\RS}] - \frac{\Tr(O)}{d} = \frac{d-1}{d(1-\epsilon)^n -1} \av{O_0} > \frac{d-1}{d(1-\epsilon)^n} \av{O_0} \geq \frac{d-1}{d}(1+\epsilon)^n \av{O_0}, 
\end{align}
where we assumed $d(1-\epsilon)^n -1 > 1$. Therefore $\bias{o_{\RS}}{O} \geq |\av{O_0}| (\frac{d-1}{d}(1 + \epsilon)^n-1 ) \geq \abs{\av{O_0}(\frac{1}{2}(1 + \epsilon)^n-1 )}$.
\end{proof}

\section{Bias mitigation conditions} \label{app: Error mitigation}
In this section, we derive sufficient conditions for the robust shadow protocol to work even in the presence of gate-dependent Pauli noise.
The ultimate aim is to determine under which conditions the estimation bias decreases, i.e., $\bias{o_{\RS}}{O} \leq  \bias{o}{O}$. 
Here, we characterize conditions of overcorrection (\cref{fig: mitigation regions}) and prove that istropic Pauli noise is well-conditioned in the sense that \acl{RSE} `works well' (\cref{prop: Offset concentration Appendix}). 
 Let $\bar \lambda = \sum_{a \neq 0} \bar \lambda_a/(d^2-1)$ be the mean of $\bar \lambda_a$ (excluding $\lambda_0$) and $O_0, \rho_0$ be the traceless parts of $O$ and $\rho$ respectively. We further define $\mc D(\vec \Delta)$ to be the diagonal channel in Pauli basis with values $\Delta_a$ on the diagonal.
The following observation provides insight onto how the bias for gate-dependent Pauli noise is determined by a single parameter $f_{\eff}$. 
\begin{obs}\label{obs: Pauli noise robust shadow bias}
 Let the average Pauli eigenvalues $\{\bar \lambda_a\}$ be parameterized as $\bar \lambda_{a} = \bar \lambda + \Delta_a$ for $a \neq 0$, and $\lambda_0 = 1$ (trace preserving noise). Then the bias is determined by the quantity $\bar \lambda + \frac{\brar{O_0}\mc D (\vec \Delta)\ketr{\rho_0}}{\langle O_0 \rangle} \eqqcolon f_{\eff}$ as
\begin{align}
\bias{o_\RS}{O} = |\langle O_0 \rangle| \cdot |1-\hat{f}_m^{-1}f_{\eff}|
\quad \text{and} \quad 
\bias{o}{O} = |\langle O_0 \rangle| \cdot |1- f_{\eff}|\,.
\end{align}
\end{obs}
Why this holds is shown in the proof of \cref{prop-app: Pauli twirled RS} below. 
The difference to the case of left gate-independent noise lies in the additional term $\frac{\brar{O_0}\mc D (\vec \Delta)\ketr{\rho_0}}{\langle O_0 \rangle}$ that quantifies how Pauli noise aligns with the signal $\braketr{O_0}{\rho_0}$. 
Perfect bias mitigation is achieved 
for
$\hat{f}_m = f_{\eff}$.
If the additional term is small, then a feasible strategy is to just estimate $\bar \lambda$ and use it as a mitigation parameter. This is essentially what robust shadow estimation does, albeit by only taking the average over $a \in \Zsf^n$ since $\hat f_m^{-1} = \EE_{a \in \Zsf^n} \bar{\lambda}_a$. 

\begin{figure}[t] 
  \includegraphics{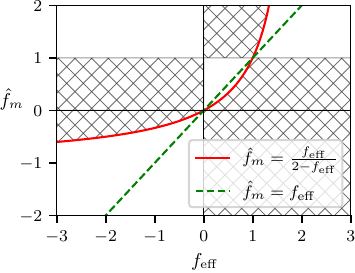}
  \caption{Parameter regions where robust shadow estimation increases the bias (hatched) and regions where it decreases the bias (plain). Perfect mitigation is achieved on the green dashed line. \label{fig: mitigation regions}
  }
\end{figure}

In \cref{fig: mitigation regions} areas in parameter space corresponding to $|\av{\hat O_{\RS}} - \langle O\rangle| \leq |\langle \hat O \rangle - \langle O\rangle|$ are visualized. 
Here, we allow $|\hat{f}_m| > 1$, which is not the case if $\hat{f}_m$ is estimated according to the robust shadow estimation protocol, but still instructive for large $\frac{\brar{O_0}\mc D (\vec \Delta)\ketr{\rho_0}}{\langle O_0 \rangle}$. 
We can see that for $f_{\eff}\leq 0$ and $f_{\eff}\geq 2$, the error is always reduced by setting $\hat{f}_m \rightarrow \infty$, meaning we return $\estimator{o}_{\mathrm{RS}} = 0$. 

This is the case for $\brar{O_0}\mc D (\vec \Delta)\ketr{\rho_0} \gg \langle O_0 \rangle$, i.e.\ large error alignment and small true expectation value $\av{O_0}$. 
In the most relevant parameter regime of $0 \leq f_{\eff} \leq 1$, we 
obtain
the condition $\hat{f} \geq \frac{f_{\eff}}{2- f_{\eff}}$ that safeguards against introducing new errors due to overcorrection. 
A formal treatment of the success criteria for robust shadow estimation is given in the following proposition. 
We denote the canonical inner product by $\innerp{\gamma}{\Delta}\coloneqq \sum_a \gamma_a^\ast \Delta_a$.

\begin{prop}\label{prop-app: Pauli twirled RS}
  Let $\hat O_{RS}$ be the robust shadow estimate on $O = \sum_a \braketr{\sn_a}{O}\ketr{\sn_a}, \, \rho = \sum_a \braketr{\sn_a}{\rho}\ketr{\sn_a}$ with mitigation parameter $f_m \in [-1,1]$ and let $\gamma_a = \braketr{O}{\sn_a} \braketr{\sn_a}{\rho}$ for $a \neq 0$. 
  Under gate-dependent Pauli noise parameterized as $\bar \lambda_a = \bar \lambda + \Delta_a$ with $\lambda_{\ii} = 1$, it holds that $\bias{o_\RS}{O} \leq \bias{o}{O}$ iff
  \begin{align}
    \left(f_{\eff} \geq 0 \land f_m \geq \frac{f_{\eff}}{2- f_{\eff}}\right) \quad \lor \quad 
    \left(f_{\eff} < 0 \land f_m \leq \frac{f_{\eff}}{2- f_{\eff}}\right)\, 
  \end{align}
  for $f_{\eff} \coloneqq \bar \lambda + \frac{\symplecticp{\vec \gamma}{\vec \Delta}}{\langle O_0 \rangle}$.
\end{prop}

\begin{proof}
The error of non-mitigated shadow estimation under Pauli noise is given in Eq.~\eqref{eq: Pauli noise bound 1}, from which we gather that 
\begin{align}
\bias{o}{O} &= \left|\sum_{a \neq 0} \braketr{O}{\sn_a} \braketr{\sn_a}{\rho} \abs{1-\lambda_a}\right| = \left|\sum_{a \neq 0} \gamma_a \abs{1-\lambda_a}\right| 
= 
\left|\innerp{\vec \gamma}{\mathbf 1 - \bar \lambda \mathbf{1} - \vec \Delta} \right|\,
\end{align}
per the assumption that $\lambda_{\ii} = 1$ and $\lambda_{a \neq 0} = \bar \lambda + \Delta_a$, 
where $\mathbf 1$ is the vector of which all entries are $1$. 
When the mitigation factor is included we obtain 
$\hat{\tilde{S}}^{-1}\tilde S = \ketbrar{\onen}{\onen} + \sum_{a \neq 0}f_m^{-1} \lambda_a \ketbrar{\sn_a}{\sn_a}$, 
and the mitigated error becomes
\begin{align}
\bias{o_\RS}{O} &= \left|\sum_{a \neq 0} \gamma_a (1-f_m^{-1}\lambda_a)\right| = \left|\innerp{\vec \gamma}{\mathbf 1 - f_m^{-1} \bar \lambda \mathbf 1 - f_m^{-1} \Delta} \right| .
\end{align}
Since $\innerp{\vec \gamma}{\mathbf 1} = \sum_{a \neq 0}\gamma_a$ is by definition the expectation value of the traceless part of the observable, $\langle O_0 \rangle$, we can simplify the error terms by looking at the relative errors:
\begin{align}
\frac{\bias{o}{O}}{|\langle O_0 \rangle|} 
&= 
\left|1- \left(\bar \lambda  + \frac{\innerp{\gamma}{\Delta}}{\langle O_0 \rangle}\right) \right| = |1-f_{\eff}|
\intertext{and}
\frac{\bias{o_\RS}{O}}{|\langle O_0 \rangle|} 
&= 
\left|1- f_m^{-1}\left(\bar \lambda  + \frac{\innerp{\gamma}{\Delta}}{\langle O_0 \rangle}\right) \right| = |1-f_m^{-1}f_{\eff}|\, .
\end{align}
To determine when $\bias{o_\RS}{O} \leq \epsilon$ holds, we have to look at four cases corresponding the signs of $f_{\eff}$ and $f_m$. We also assume that $|f_m| \leq 1$ since $|\bar \lambda| \leq 1$ for any physical noise model.\\

\begin{enumerate}[(i)]
 \item  \textit{$f_{\eff}$ and $f_m$ are of opposite sign:} \\
In this case we have that $f_m^{-1} f_{\eff} \leq 0$ and therefore 
\begin{align}
  \bias{o}{O} = |1-f_{\eff}| \leq 1+|f_{\eff}| \leq 1+|f_m^{-1} f_{\eff}| = |1-f_m^{-1} f_{\eff}| = \bias{o_\RS}{O}
\end{align}
and the error mitigation technique always increases the error. 

 \item \textit{$f_{\eff} \leq 0 $ and $ f_m \leq 0$:}\\
 Let $f_m \leq \frac{f_{\eff}}{2-f_{\eff}}$, then $f_m^{-1} \leq \frac{2-f_{\eff}}{f_{\eff}}$ and $-f_m^{-1}f_{\eff} \leq f_{\eff}-2$. 
 Therefore, $\bias{o_\RS}{O} = |1-f_m^{-1}f_{\eff}| \leq |-1 + f_{\eff}| = \bias{o}{O}$ and error mitigation is achieved. 
 The condition is tight, since $f_m > \frac{f_{\eff}}{2-f_{\eff}}$ leads in the same fashion to $\bias{o_\RS}{O} \geq \bias{o}{O}$.

 \item \textit{$f_{\eff} > 0$ and $f_m > 0$:}\\
 Let now $f_m \geq \frac{f_{\eff}}{2-f_{\eff}}$. 
 We obtain $f_m^{-1} \geq \frac{2-f_{\eff}}{f_{\eff}}$ and $-f_m^{-1}f_{\eff} \leq f_{\eff}-2$, as $-f_{\eff}$ is negative. 
 Thus, again $\bias{o_\RS}{O} = |1-f_m^{-1}f_{\eff}| \leq |-1 + f_{\eff}| = \bias{o}{O}$ and errors are mitigated. The condition is again tight, as $f_m > \frac{f_{\eff}}{2-f_{\eff}}$ leads to $\bias{o_\RS}{O} \geq \bias{o}{O}$.
\end{enumerate}
\end{proof} 

To make a statement about $f_{\eff}$, without assuming explicit knowledge of the noise model and the prepared state, we treat $\Delta$ as a random vector. 
The following lemma gives a concentration inequality for the well known fact that in high dimensions, a uniformly distributed random vector on the unit sphere is almost orthogonal to any given fixed vector with high probability. 

\begin{lem}[Adapted from \cite{Vazirani11LectureNotes}]\label{lem: Random vector orthogonality}
For an arbitrary normalized vector $\vec x \in \mathbb{R}^k$ and a random vector $\vec g$ which is uniformly distributed on the unit sphere, holds that $\PP(|\innerp{\vec x}{\vec g}| \geq \frac{t}{\sqrt{k-1}}) \leq \e^{-t^2/2}$\,.
\end{lem}
\begin{proof}
If $g$ is uniformly distributed, then also the random vectors $O\vec g$ for $O \in \mathrm{SO}(k)$ are distributed uniformly. 
This implies that $\PP(|\innerp{\vec x}{\vec g}| \geq \frac{t}{\sqrt{k-1}})$ is independent of $\vec x$ and w.l.o.g.\ we choose $\vec x$ to be the first canonical basis vector, $\vec x = \vec{e}_1$.

The surface area of a $k$-dimensional unit sphere with radius $r$ is given by $A_S(k,r) = \frac{2 \pi^{k/2} r^{k-1}}{\Gamma(\frac{k}{2})}$. 
Since $\vec g$ is uniformly distributed on the unit sphere, the probability that $\abs{\innerp{e_1}{g}} = \abs{g_1}$ is larger than $\frac{t}{\sqrt{k-1}}$ is given by the ratio of the surface area of two spherical caps of height $1 - \frac{t}{\sqrt{k-1}}$ to the surface area $A_{S}(k,1)$ of the unit sphere. The base of these caps has radius $a = \sqrt{1-\frac{t^2}{k-1}}$, and we can bound the surface area of the two caps by the surface are of the full sphere of radius $a$ as $2A_{\mathrm{CAP}}(k,a) \leq A_S(k,a)$. This leads us to the bound
\begin{align}
  \PP\left[|\innerp{\vec{e}_1}{\vec g}| \geq \frac{t}{\sqrt{k-1}}\right] 
  = 
  \frac{2 A_{\mathrm{CAP}}(k,r)}{A_S(k,1)} 
  \leq 
  \frac{A_S(k,a)}{A_S(k,1)} 
  = 
  \left(1 - \frac{t^2}{k-1}\right)^{\frac{k-1}{2}} \leq \e^{-t^2/2}.
\end{align}
\end{proof}
We choose this bound for simplicity, slightly stronger versions can be found in e.g.\ \citet{Dasgupta2003} and references therein. 
An example of uniformly distributed vectors on the sphere is given by normalized standard Gaussian vectors where each entry is drawn from the 0-mean and unit variance normal distribution $\mc N(0,1)$. 

If the term $\frac{\innerp{\gamma}{\Delta}}{\langle O_0 \rangle}$ in $f_{\eff} = \bar \lambda  + \frac{\innerp{\gamma}{\Delta}}{\langle O_0 \rangle}$ is small, then it is sufficient to estimate $\bar \lambda$ to correct the bias. In \Cref{prop: Offset concentration Appendix} below, we show that this is the case under reasonable assumptions.   

\begin{prop} \label{prop: Offset concentration Appendix}
  Let $\langle O_0 \rangle \geq C_1$ and $\twonorm{O} < C_2$. If $\vec \Delta$ is a random vector such that $\vec \Delta/\lpnorm[2]{\vec \Delta}$ is uniformly distributed on the unit sphere and $\lpnorm[2]{\vec \Delta} \leq \mc O(d)$ with high probability,
  then 
  $f_{\eff} = \EE[\hat f_m] +  \LandauO(1/\sqrt{d})$ with high probability.
\end{prop}

\begin{proof}
If $\Delta$ is uniformly distributed on the unit sphere, then its entries $\Delta_a$ are zero mean random variables. 
Therefore, 
\begin{equation}
  \EE[\hat f_m] = \bar \lambda + \frac{1}{d-1}\sum_{a \in \Zsf^n\backslash \ii} \EE[\bar \lambda_a] 
  = 
  \bar \lambda + \frac{1}{d-1}\sum_{a \in \Zsf^n\backslash \ii} \EE[\Delta^a] 
  = 
  \bar \lambda
\end{equation}
and it remains to bound $\frac{\innerp{\gamma}{\Delta}}{\langle O_0 \rangle}$. 
We write the inner product as $\innerp{\vec \gamma}{\vec \Delta} = \lpnorm[2]{\vec \gamma} \lpnorm[2]{\vec \Delta}\, \innerp{\n \gamma}{\n \Delta}$ with $\n \gamma,\, \n \Delta$ normalized. 
Then we know from \cref{lem: Random vector orthogonality} for $k = d^2-1$ that 

\begin{align}
\PP\left[|\innerp{\n \gamma}{\n \Delta}| \geq \frac{t}{\sqrt{d^2-2}}\right] \leq \e^{-t^2/2}.
\end{align}
From the assumption that $\lpnorm[2]{\vec \Delta} \leq \mc O(d)$ and the union bound leads to $\lpnorm[2]{\vec \Delta}|\innerp{\n \gamma}{\n \Delta}| \leq \mc O(1)$ with high probability. 
Since $\abs{\braketr{\rho}{\sn_a}} \leq 1/\sqrt{d}$ for any quantum state $\rho$, we further have that $\lpnorm[2]{\vec \gamma} \leq \twonorm{O}/\sqrt{d} \leq C/\sqrt{d}$ and thus $\frac{\innerp{\gamma}{\Delta}}{\langle O_0 \rangle} = \mc{O}(1/\sqrt{d})$ with high probability. 
\end{proof}

\Cref{prop: Offset concentration Appendix} formalizes the intuition that in large systems noise is unlikely to be malicious, i.e.\ a randomly distributed noise vector $\vec \Delta$ is unlikely to align with the signal $\gamma$.

The Pauli eigenvalue average $\bar \lambda$ can also be related to the average gate fidelity of the frame operator $\tilde S$. 
Using the know relation $F_{\mathrm{Avg}}(\tilde S) = (d^{-1}\Tr[\tilde S] + 1)/d+1$ and $\Tr[\tilde S] = 1+\sum_{a \neq 0}\lambda_a/(d+1)$ \cite{Niel02}. 
A quick rearrangement yields $\bar \lambda = (d F_{\mathrm{Avg}}(\tilde S))(d+1)/(d-1)$, suggesting that an estimate of the average gate fidelity $F_{\mathrm{Avg}}(\tilde S)$ would also give us an estimate of $\bar \lambda$.

For a detailed bound in terms of a given error probability $\delta$, we now look to the example of a Gaussian random noise vector. 
This random vector provides a concrete example for a noise distribution that satisfies the requirements of \cref{prop: Offset concentration Appendix}. 

\begin{prop}\label{prop:isotropic_noise}
  Let $\vec \Delta$ be a Gaussian random vector of length $k = d^2-1$ with i.i.d. entries from $\mc N(0,\sigma^2)$ and $\vec \gamma$ be an arbitrary real vector of the same dimension.
  It holds that 
  \begin{align*}
    \PP\left[|\innerp{\vec \gamma}{\vec \Delta}| \leq \sigma^2\lpnorm[2]{\vec \gamma}g(\delta)\sqrt{\frac{k}{k-1}}\left(1 + \frac{g(\delta)}{\sqrt{k}} + \frac{g^2(\delta)}{k}\right)\right] 
     \geq 1-\delta \,,
  \end{align*}
  where $g(\delta) = \sqrt{\log(2/\delta)}$.
\end{prop}

\begin{proof}
We begin with a bound for $\lpnorm[2]{\vec \Delta}^2/\sigma^4 = \lpnorm[2]{X}^2 = \sum_{i = 1}^{k} x_i^2$, where now $x_i \sim \mc N(0,1)$ and therefore $\lpnorm[2]{X}^2$ is a Chi-squared distributed random variable. 
We can then use the Laurent-Massart inequality which states that $\PP(\lpnorm[2]{X}^2 - k \geq 2\sqrt{kt} + 2t )\leq \e^{-t}$. 
We hence obtain $\PP(\lpnorm[2]{X} \geq \sqrt{k}\sqrt{1 + 2\sqrt{t/k} + 2t/k})\leq \e^{-t}$, and we can use $\sqrt{1 + 2\sqrt{t/k} + 2t/k} \leq 1 + \sqrt{t/k} + t/k$. 
By fixing the failure probability to $\e^{-t} = \delta/2$ we arrive at 
\begin{equation}
  \PP\left[\lpnorm[2]{X} \geq \sqrt{k}(1 + \sqrt{\frac{\mathrm{log}(2/\delta)}{k}} + \frac{\mathrm{log}(2/\delta)}{k})\right]\leq \frac{\delta}{2}
\end{equation}
and 
\begin{equation}\label{eq:concentration_Delta}
  \PP\left[\lpnorm[2]{\vec \Delta} \geq \sigma^2\sqrt{k}(1 + \frac{g(\delta)}{\sqrt{k}} + \frac{g^2(\delta)}{k})\right]\leq \frac{\delta}{2}. 
\end{equation}

Our aim is to bound the inner product $\innerp{\vec \gamma}{\vec \Delta}$, which we will write as $\lpnorm[2]{\vec \gamma} \lpnorm[2]{\vec \Delta}\, \innerp{\n \gamma}{\n \Delta}$, with $\n \gamma,\, \n \Delta$ normalized. This allows us to use \Cref{lem: Random vector orthogonality} for $\e^{-t^2/2} = \delta/2$, and we arrive at
\begin{align}\label{eq: random vector intermediate bound}
  \PP\left[|\innerp{\vec \gamma}{\vec \Delta}| \geq \lpnorm[2]{\vec \gamma} \lpnorm[2]{\vec \Delta} \frac{g(\delta)}{\sqrt{k-1}}\right]\leq \delta/2 \,.
\end{align}
Combining this bound with the bound \eqref{eq:concentration_Delta} for $\lpnorm[2]{\vec \Delta}$ via the union bound completes the proof.  
\end{proof}


\begin{acronym}[POVM]
\acro{ACES}{averaged circuit eigenvalue sampling}
\acro{AGF}{average gate fidelity}
\acro{AP}{Arbeitspaket}

\acro{BK}{Bravyi-Kitaev}
\acro{BOG}{binned outcome generation}

\acro{CNF}{conjunctive normal form}
\acro{CP}{completely positive}
\acro{CPT}{completely positive and trace preserving}
\acro{cs}{computer science}
\acro{CS}{compressed sensing} 
\acro{ctrl-VQE}{ctrl-VQE}

\acro{DAQC}{digital-analog quantum computing}
\acro{DD}{dynamical decoupling}
\acro{DFE}{direct fidelity estimation} 
\acro{DFT}{discrete Fourier transform}
\acro{DL}{deep learning}
\acro{DM}{dark matter}

\acro{FFT}{fast Fourier transform}

\acro{GS}{ground-state}
\acro{GST}{gate set tomography}
\acro{GTM}{gate-independent, time-stationary, Markovian}
\acro{GUE}{Gaussian unitary ensemble}

\acro{HOG}{heavy outcome generation}

\acro{irrep}{irreducible representation}

\acro{JW}{Jordan-Wigner}

\acro{LBCS}{locally-biased classical shadow}
\acro{LDPC}{low density partity check}
\acro{LP}{linear program}

\acro{MAGIC}{magnetic gradient induced coupling}
\acro{MAX-SAT}{maximum satisfiability}
\acro{MBL}{many-body localization}
\acro{MIP}{mixed integer program}
\acro{ML}{machine learning}
\acro{MLE}{maximum likelihood estimation}
\acro{MPO}{matrix product operator}
\acro{MPS}{matrix product state}
\acro{MS}{M{\o}lmer-S{\o}rensen}
\acro{MSE}{mean squared error}
\acro{MUBs}{mutually unbiased bases} 
\acro{mw}{micro wave}

\acro{NISQ}{noisy and intermediate scale quantum}

\acro{ONB}{orthonormal basis}
\acroplural{ONB}[ONBs]{orthonormal bases}

\acro{POVM}{positive operator valued measure}
\acro{PQC}{parametrized quantum circuit}
\acro{PSD}{positive-semidefinite}
\acro{PSR}{parameter shift rule}
\acro{PVM}{projector-valued measure}

\acro{QAOA}{quantum approximate optimization algorithm}
\acro{QC}{quantum computation}
\acro{QEC}{quantum error correction}
\acro{QFT}{quantum Fourier transform}
\acro{QM}{quantum mechanics}
\acro{QML}{quantum machine learning}
\acro{QMT}{measurement tomography}
\acro{QPT}{quantum process tomography}
\acro{QPU}{quantum processing unit}
\acro{QUBO}{quadratic binary optimization}
\acro{QWC}{qubit-wise commutativity}

\acro{RB}{randomized benchmarking}
\acro{RBM}{restricted Boltzmann machine}
\acro{RDM}{reduced density matrix}
\acro{rf}{radio frequency}
\acro{RIC}{restricted isometry constant}
\acro{RIP}{restricted isometry property}
\acro{RMSE}{root mean squared error}
\acro{RSE}{robust shadow estimation}

\acro{SDP}{semidefinite program}
\acro{SFE}{shadow fidelity estimation}
\acro{SIC}{symmetric, informationally complete}
\acro{SPAM}{state preparation and measurement}
\acro{SPSA}{simultaneous perturbation stochastic approximation}

\acro{TT}{tensor train}
\acro{TM}{Turing machine}
\acro{TV}{total variation}

\acro{VQA}{variational quantum algorithm}
\acro{VQE}{variational quantum eigensolver}

\acro{XEB}{cross-entropy benchmarking}

\end{acronym}

\bibliographystyle{./myapsrev4-2}

\begin{thebibliography}{37}
\bibitem [{Pau()}]{Pauli-noise-assumption}%
  \BibitemOpen
  \href@noop {} {}\bibinfo {note} {Note that for (left) Pauli-invariant
  ensembles, such as local or global Clifford unitaries, gate-independent left
  noise can be ef- fectively Pauli-twirled and hence the Pauli noise assumption
  is implicit.}\BibitemShut {Stop}%
\end{thebibliography}
%

\end{document}